\documentclass[11pt]{article}
\usepackage{authblk}
\usepackage{fullpage,times}
\usepackage[utf8]{inputenc}
\usepackage{mathtools,amsmath,amsfonts,amsthm}
\usepackage{tikz}
\usetikzlibrary{shapes}
\usepackage{latexsym}
\usepackage{hyperref}

\newcommand{\G}{\mathcal{G}}
\newcommand{\E}{\mathcal{E}}
\newcommand{\F}{\mathcal{F}}
\newcommand{\HH}{\mathcal{H}}
\newcommand{\LG}{\mathcal{LG}}
\newcommand{\ELG}{\mathcal{LG}^\mathrm{ext}}
\newcommand{\cS}{\mathcal{S}}
\newcommand{\V}{\mathcal{V}}
\newcommand{\YY}{{Y}}

\newcommand{\ZZ}{{Z}}
\newcommand{\EE}{{\E}}

\newcommand{\ZO}{\{0,1\}}
\def\B{\{0,1\}}
\newcommand\R{{\mathbb R}}

\newcommand{\Learn}{\mathsf{Load}}
\newcommand{\Load}{\mathsf{Load}}
\newcommand{\Triangle}{\mathsf{Triangle}}
\newcommand{\Loadall}{\mathsf{DenseLoad}}
\newcommand{\Loadsparse}{\mathsf{SparseLoad}} %{S}\mathrl{parse\mathcal{L}oad}

\newcommand{\Advpm}{\mathrm{Adv}^\pm}
\newcommand{\GJohnson}{\G_\mathrm{Johnson}}
\newcommand{\GOR}{\G_{\mathrm{OR}}}
\newcommand{\Exp}{\mathop{\mathrm{Exp}}}
\newcommand{\red}[1]{{\color{red} #1}}

\newtheorem{theorem}{Theorem}[section]
\newtheorem{lemma}[theorem]{Lemma}
\newtheorem{definition}[theorem]{Definition}
\makeatletter
\def\env@cases{%
  \let\@ifnextchar\new@ifnextchar
  \left\lbrace
  \def\arraystretch{1.2}%
  \array{l@{\quad}l@{}}% Formerly @{}l@{\quad}l@{}
}
\makeatother

\date{}

\title{Extended Learning Graphs for Triangle Finding\footnote{This work has been partially supported by 
the European Commission project Quantum Algorithms (QALGO) and
    the French ANR Blanc project RDAM.}}
\author[1]{Titouan Carette}
\author[2]{Mathieu Lauri\`ere}
\author[3]{Fr\'ed\'eric Magniez}
\affil[1]{Ecole Normale Supérieure de Lyon, Lyon, France}
\affil[2]{%CNRS, IRIF, Univ Paris Diderot, Paris, France,\\ \&\, 
NYU-ECNU Institute of Mathematical Sciences at NYU Shanghai, China}
\affil[3]{CNRS, IRIF, Univ Paris Diderot, Paris, France}
\begin{document}

\maketitle

\begin{abstract}
We present new quantum algorithms for Triangle Finding improving its best previously known quantum query complexities for both dense and spare instances.
For dense graphs on $n$ vertices, we get a query complexity of $O(n^{5/4})$ without any of the extra logarithmic factors  present in the previous algorithm of Le Gall [FOCS'14].
For sparse graphs with $m\geq n^{5/4}$  edges, we get a query complexity of $O(n^{11/12}m^{1/6}\sqrt{\log n})$, which is better than the one obtained by Le Gall and Nakajima [ISAAC'15] when $m \geq n^{3/2}$.
We also obtain an algorithm with query complexity ${O}(n^{5/6}(m\log n)^{1/6}+d_2\sqrt{n})$ where $d_2$ is the variance of the degree distribution.

Our algorithms are designed and analyzed in a new model of learning graphs that we call extended learning graphs. In addition, we present a framework in order to easily combine and analyze them. 
As a consequence we get much simpler algorithms and analyses than previous algorithms of Le Gall {\it et al} based on the MNRS quantum walk framework [SICOMP'11].
\end{abstract}

%\tableofcontents

% SSSSSSSSSSSSSSSSSSSSSSSSSSSSSSSSSSSSSSSSSSSSSSSSSSSSSSSSSSSSSSSSSSSSSSSSSSSSSSSSSSSSSSSSSSSSSSSSSSSSSSSSSSSSSSSSSSSSSSSSSSSSSSSSSSSSSSSSSSSS %
\section{Introduction}
% SSSSSSSSSSSSSSSSSSSSSSSSSSSSSSSSSSSSSSSSSSSSSSSSSSSSSSSSSSSSSSSSSSSSSSSSSSSSSSSSSSSSSSSSSSSSSSSSSSSSSSSSSSSSSSSSSSSSSSSSSSSSSSSSSSSSSSSSSSSS %
Decision trees form a simple model for computing Boolean functions by
successively reading the input bits until the value of the function
can be determined. In this model, the \emph{query complexity} is the number of 
 input bits queried.
This allows us to study the complexity of a
function in terms of its structural properties. For instance, sorting an array of size $n$
can be done using $O(n\log n)$ comparisons, and this is optimal for comparison-only algorithms.

In an extension of the deterministic model, one can also allow randomized and even quantum computations.
Then the speed-up can be exponential for partial functions ({\it i.e.} problems with promise) when we compare 
deterministic with randomized computation, and randomized with quantum computation.
The case of total functions is rather fascinating. For them, the best possible gap can only be polynomial
between each models~\cite{nis91,bbcmw01}, which is still useful in practice for many problems.
But surprisingly, the best possible gap is still an open question, event if it was improved for both models  very recently~\cite{abblss16,abk16}.
%Nonetheless, for total functions, the speed-up provided by randomized computations can only be polynomial~\cite{nis91},
%but the best possible gap is still open.
%However, after thirty years of very little progress, 
%it has been significantly improved~\cite{abblss16}.
In the context of quantum computing, query complexity %has been also extensively studied.
%is a black-box model of quantum computation, where the resource measured is the number of queries to the input needed to compute a function.  
%There, this model 
captures the great algorithmic successes 
of quantum computing like the search algorithm of Grover~\cite{gro96} and the period finding subroutine of Shor's 
factoring algorithm~\cite{shor97}, while at the same time it is simple enough that one can often show tight lower bounds.
%Again, for total functions, the speed-up is at most polynomial~\cite{bbcmw01} and the best possible separation has been improved very recently~\cite{abk16}.

Reichardt~\cite{rei11} showed that the general adversary bound, formerly just a lower bound technique 
for quantum query complexity~\cite{hls07}, is also an upper bound.  This characterization 
has opened an avenue for designing quantum query algorithms.  However, even for simple functions it is challenging to 
find an optimal bound.
Historically, studying the query complexity of specific functions led to amazing progresses in our understanding of quantum computation, by providing new algorithmic concepts and tools for analyzing them.
Some of the most famous problems in that quest are Element Distinctness and Triangle Finding~\cite{bdhhmsw05}.
Element Distinctness consists in deciding if a function takes twice the same value on a domain of size $n$,
whereas Triangle Finding consists in determining if an $n$-vertex graph has a triangle. 
%A variant of Element Distinctness is $k$-Element Distinctness for which we are now checking if the function takes $k$ times the same value.
Quantum walks were used to design algorithms with optimal query complexity for Element Distinctness. Later on,
 a general framework for designing quantum walk based algorithms was developed with various applications~\cite{mnrs11}, including for Triangle Finding~\cite{mss07}.

For seven years, no progress on Triangle Finding was done until Belovs developed his beautiful
model of \emph{learning graphs}~\cite{bel12}. Learning graphs can be
viewed as the minimization form of the general adversary bound with an additional structure imposed on the form 
of the solution.  This additional structure makes learning graphs easier to reason about
without any background on quantum computing.
% by ensuring that the constraints are {automatically} satisfied, leaving one to focus on optimizing the objective value. 
On the other hand, they may not provide always optimal algorithms.
Learning graphs have an intuitive interpretation in terms of electrical networks~\cite{bcjkm13}.
%which make them easy to manipulate 
Their complexity is directly connected to the total conductance of the underlying network and its effective resistance. Moreover this characterization leads to a generic quantum implementation which is time efficient and preserves  query complexity.

Among other applications, learning graphs have been used to design an algorithm for Triangle Finding  with query complexity $O(n^{35/27})$~\cite{spanCert}, 
improving on the previously known bound $\tilde{O}(n^{1.3})$ obtained by a quantum walk based algorithm~\cite{mss07}.
Then the former was improved by another learning graph using $O(n^{9/7})$ queries~\cite{lms15}.
This learning graph has been proven optimal for the original class of learning graphs~\cite{BR13}, known as \emph{non-adaptive learning graphs}, for which the conductance of each edge is constant.
%Indeed,  non-adaptive learning graphs for triangle finding also works in a more general {\em weighted} setting.  
%Thus, for non-adaptive learning graphs, Triangle Finding has the same query complexity than Triangle Sum (given a weighted graph, determine if there is a triangle whose edge labels sum to $0$) 
%whose query complexity is exactly~$\Omega(n^{9/7})$.
Then, Le Gall showed that quantum walk based algorithms are indeed stronger than non-adaptive learning graphs
for Triangle Finding by constructing a new quantum algorithm with query complexity~$\tilde{O}(n^{5/4})$~\cite{gal14}.
His algorithm combines in a novel way combinatorial arguments on graphs with quantum walks. One of the key ingredient is the use of an algorithm due to Ambainis for implementing Grover Search in a model whose queries may have variable complexities~\cite{amb10}. 
Le Gall used this algorithm
% of Ambainis 
to average the complexity of different branches of its quantum walk
in a quite involved way.
%, since the framework of~\cite{mnrs11} are not made for such an averaging.
In the specific case of sparse graphs, those ideas have also demonstrated their advantage for Triangle Finding on previously known algorithms~\cite{ls15}.

The starting point of the present work is to investigate a deeper understanding of learning graphs and their extensions.
Indeed, various variants have been considered without any unified and intuitive framework.
For instance, the best known quantum algorithm for $k$-Element Distinctness 
(a variant of Element Distinctness where we are now checking if the function takes $k$ times the same value)
has been designed by several clever
relaxations of the model of learning graphs~\cite{bel12}.  Those relaxations led to algorithms 
more powerful than non-adaptive learning graphs, but at the price of a more complex and less intuitive analysis.
In \textbf{Section~\ref{sec:extlg}}, we extract several of those concepts that we formalize in our 
new model of \emph{extended learning graphs} (\textbf{Definition~\ref{def:extLGfct}}).
We prove that their complexity (\textbf{Definition~\ref{extlg}}) is always an upper bound on the query complexity of the best quantum algorithm solving the same problem (\textbf{Theorem~\ref{thm:ELG-querycomplexity}}).
We also introduce the useful notion of \emph{super edge} (\textbf{Definition~\ref{def:super}}) for compressing some given portion of a learning graph. We use them to encode efficient learning graphs querying a part of the input on some given index set (\textbf{Lemmas~\ref{lem:complexityLearnAll} and~\ref{lem:complexityLearn}}).
In some sense, we transpose to the learning graph setting the strategy of 
finding all $1$-bits of some given sparse input using Grover Search.
%with a complexity depends on its Hamming weight. % of the input.
%The resulting complexity depends on the number of such bits.

In \textbf{Section~\ref{subsec:decomp-lemmas}}, we provide several tools for composing our learning graphs.
We should first remind the reader that, since extended learning graphs cover a restricted class of quantum algorithms,
it is not possible to translate all quantum algorithms in that model. Nonetheless we succeed
%transfer to our model of extended learning graphs 
for two important algorithmic techniques:
Grover Search with variable query complexities~\cite{amb10} (\textbf{Lemma~\ref{groverExtended}}),
and Johnson Walk based quantum algorithms~\cite{mss07,mnrs11} (\textbf{Theorem~\ref{johnsonExtended}}).
In the last case, we show how to incorporate the use of super~edges for querying sparse inputs.

We validate the power and the ease of use of our framework on Triangle Finding in \textbf{Section~\ref{sec:triangle}}.
First, denoting $n$ is the number of vertices, we provide a simple adaptive learning graph with query complexity $O(n^{5/4})$, 
whose analysis is arguably much simpler than the algorithm of Le Gall, and whose complexity is cleared of logarithmic factors (\textbf{Theorem~\ref{thm:dense}}).
This also provides the first natural separation between non-adaptive and adaptive learning graphs.
Then, we focus on sparse input graphs and develop extended learning graphs. All algorithms of~\cite{ls15} could be rephrased in our model. But more importantly, we show that one can design more efficient ones.
For sparse graphs with $m\geq n^{5/4}$  edges, we get a learning graph with query complexity ${O}(n^{11/12}m^{1/6}\sqrt{\log n})$, which improves the results of~\cite{ls15} when $m \geq n^{3/2}$ (\textbf{Theorem~\ref{thm:sparse}}).
We also construct another learning graph %whose complexity is parametrized 
with query complexity ${O}(n^{5/6}(m\log n)^{1/6}+d_2\sqrt{n})$,
where $d_2$ is the variance  of the degree distribution
(\textbf{Theorem~\ref{thm:sparsenew}}).
To the best of our knowledge, this is the first quantum algorithm for Triangle Finding whose complexity depends on this parameter $d_2$.

% SSSSSSSSSSSSSSSSSSSSSSSSSSSSSSSSSSSSSSSSSSSSSSSSSSSSSSSSSSSSSSSSSSSSSSSSSSSSSSSSSSSSSSSSSSSSSSSSSSSSSSSSSSSSSSSSSSSSSSSSSSSSSSSSSSSSSSSSSSSS %
\section{Preliminaries}\label{prelim}
% SSSSSSSSSSSSSSSSSSSSSSSSSSSSSSSSSSSSSSSSSSSSSSSSSSSSSSSSSSSSSSSSSSSSSSSSSSSSSSSSSSSSSSSSSSSSSSSSSSSSSSSSSSSSSSSSSSSSSSSSSSSSSSSSSSSSSSSSSSSS %
%\subsection{Background}
%For any integer $q\geq 1$, let $[q] = \{1, 2,\ldots, q\}$.
We will deal with Boolean functions  of the form 
$f : \ZZ \to \B$, where $\ZZ\subseteq \B^N$.
 %where the input to the function can be thought of as the complete undirected graph 
 %(possibly with self-loops) 
%on vertex set $[n]$.
%, whose edges are colored by elements from $[q]$. When $q=2$, the input is of course just a directed graph (again possibly with self-loops).
In the query model, given a function $f : \ZZ \to \B$, the goal is to evaluate $f(z)$ by making 
as few queries to the $z$ as possible.  
A query is a question of the form `What is the value of $z$ in position 
$i\in [N]$?', to which is returned $z_i\in\B$. 

In this paper we will discuss functions whose inputs are themselves graphs. 
Then $z$ will encode an undirected graph $G$ on vertex set $[n]$, that is $N={n \choose 2}$ in order to encode the possible edges of $G$. Then $z_{ij}=1$ iff $ij$ is an edge of $G$.
 %To prevent confusion we will refer to vertices and edges of the learning graph as \emph{$L$-vertices} and \emph{$L$-edges} respectively.

 In the quantum query model, these queries can be 
asked in superposition.  We refer the reader to the survey~\cite{hs05} for precise definitions and background on
the quantum query model.  
We denote by $Q(f)$ the number of queries needed by a quantum algorithm to 
evaluate $f$ with error at most $1/3$.
Surprisingly, the general adversary bound, that we define below, is a tight characterization of $Q(f)$.
\begin{definition}
\label{def:advpm-bound-sdp}
	Let $f : \ZZ \to \ZO$ be a function, with $\ZZ \subseteq \B^N$. The \emph{general adversary bound} $\Advpm(f)$ is defined as the optimal value of the following optimization problem:
$$ \text{minimize: }  \max_{z \in \ZZ} \sum_{j \in [n]} X_j [z,z]  \quad
\text{subject to: } 
   		\begin{aligned}[t]
                   \ \sum_{\mathclap{j \in [n] \,:\, x_j \neq y_j}}  X_j [x,y] & = 1, \ \text{when  $f(x) \neq f(y)$},
                    \\[-2ex]
                  X_{j} & \succeq 0,  \ \forall  j \in [N],
                \end{aligned}
$$
%  \begin{alignat*}{2}
%  & \text{minimize: } & & \max_{z \in \ZZ} \sum_{j \in [n]} X_j [z,z]  \\
%   & \text{subject to: }& \quad & \sum_{\mathclap{j \in [n] \,:\, x_j \neq y_j}} \; 
%   		\begin{aligned}[t]
%                    X_j [x,y] & = 1,& \hbox{ whenever } &  f(x) \neq f(y),
%                    \\[1ex]
%                  X_{j} & \succeq 0, & \forall & \; j \in [N],
%                \end{aligned}
%\end{alignat*}
  where the optimization is over positive semi-definite matrices $X_j$ with rows and columns labeled by the elements of $\ZZ$, and $X_j[x,y]$ is used to denote the $(x,y)$-entry of $X_j$.
%  element of matrix $X$ on the intersection of the row and column labeled by $x$ and $y$, respectively.
\end{definition}
\begin{theorem}[\cite{hls07,lmrss11,rei11}]\label{qadv}
	$Q(f) = \Theta(\Advpm(f))$.
\end{theorem}

\section{Extended learning graphs}\label{sec:extlg}
% SSSSSSSSSSSSSSSSSSSSSSSSSSSSSSSSSSSSSSSSSSSSSSSSSSSSSSSSSSSSSSSSSSSSSSSSSSSSSSSSSSSSSSSSSSSSSSSSSSSSSSSSSSSSSSSSSSSSSSSSSSSSSSSSSSSSSSSSSSSS %
Consider some fixed Boolean function $f : \ZZ \to \B$, where $\ZZ\subseteq \B^N$.
The set of positive inputs (or instances) will be usually denoted by $\YY=f^{-1}(1)$.
Remind that a \emph{$1$-certificate} for $f$ on $y\in\YY $ is a subset $I\subseteq [N]$ of indices such that $f(z)=1$ for
 every  $z\in\ZZ$ with $z_I=y_I$. % satisfies $f(z)=1$.
%By convention, we will most often call $\ZZ\subseteq [q]^N$ the set of inputs (or instances) and $\YY$ (resp. $\ZZ \backslash \YY$) the set of positive (resp. negative) inputs.

%By convention, we will most often call $\ZZ\subseteq \B^N$ the set of inputs (or instances) and $\YY$ (resp. $\ZZ \backslash \YY$) the set of positive (resp. negative) inputs.
%We now formally define an adaptive learning graph and its complexity.

% subsubsubsubsubsubsubsubsubsubsubsubsubsubsubsubsubsubsubsub %
\subsection{Model and complexity}
Intuitively, learning graphs are simply electric networks of a special type.
The network is embedded in a rooted directed acyclic graph, which
has few similarities with decision trees. %, but also many differences.
Vertices are labelled by subsets $S\subseteq [N]$ of indices. %input positions.
Edges are basically from any vertex labelled by, say, $S$ to any other one labelled $S\cup \{j\}$,
for some $j\not\in S$.
Such an edge can be interpreted as querying the input bit $x_j$, while $x_S$ has been previously learnt. The weight on the edge is its conductance: the larger it is, the more flow will go through it. Sinks of the graph are labelled by potential $1$-certificates of the function we wish to compute.

Thus a random walk on that network starting from the root (labelled by $\emptyset$),
with probability transitions proportional to conductances, will hit a $1$-certificate
with average time proportional to the product of the total conductance by the effective resistance between 
the root of leaves having $1$-certificates~\cite{bcjkm13}.
%The orientation of edges is only meant for 

If weights  are independent of the input, then the learning graph is called \emph{non-adaptative}.
When they depend on previously learned bits, it is \emph{adaptative}.
But in quantum computing, we will see that they can also depend on both the value of the next queried bit and the value of the function itself! We call them \emph{extended learning graphs}.

Formally, we generalize the original model of learning graphs by allowing 
two possible weights on each edge: one for positive instances and one for negative ones. 
Those weights are linked together as explained in the following definition.

%more flexibility in the weight function than the adaptive LG model. 
%For this we replace the weight function by a pair of functions. They are both more general but are also linked together. 
\begin{definition}[Extended learning graph]\label{def:extLGfct}
Let $\YY\subseteq \ZZ$ be finite sets.
An \emph{extended learning graph} $\G$ is $5$-tuple $(\V, \E, \cS, \{w_z^b : z\in \ZZ,b\in\B\}, \{p_y: y \in Y\})$ satisfying
\begin{itemize}
\item $(\V,\E)$ is a directed acyclic graph rooted in some vertex $r\in \V$;
\item $\cS$ is a vertex labelling mapping each  $v \in \V$ to  $\cS(v)\subseteq [N]$ such that
$\cS(r) = \emptyset$ and $\cS(v) = \cS(u)\cup\{j\}$
for every  $(u,v)\in\E$ and 
  some $j\not\in\cS(u)$;
\item Values $w_z^b(u,v)$ are in $\R_{\geq 0}$ and
 depend on $z$ only through $z_{\cS(v)}$, for every $(u,v)\in\E$;
\item $w_x^0(u,v)=w_y^1(u,v)$ for all $x \in \ZZ\setminus\YY, y \in\YY$ and edges $ (u,v) \in \EE$ such that
$x_{\cS(u)}=y_{\cS(u)}$ and $x_{j}\neq y_{j}$ with $\cS(v)=\cS(u)\cup\{j\}$.
\item  $p_y : \E \rightarrow \mathbb{R}_{\geq 0}$ is a unit flow whose source is the root and
such that $p_y(e) = 0$ when $w_y^1(e) = 0$,
for every $y \in \YY$.
\end{itemize}
We say that $\G$ is a \emph{learning graph for} some function $f : \ZZ \rightarrow \B$, when 
$Y=f^{-1}(1)$ and each sink of $p_y$ contains a 
$1$-certificate for $f$ on $y$, for all positive input $y \in f^{-1}(1)$.
\end{definition}

We also say that $\G$ is an \emph{adaptive learning graph} when $w_z^0=w_z^1$ for all $z\in \ZZ$.
If furthermore $w_z^0$ is independent of $z$, we say that $\G$ is a \emph{non-adaptive learning graph}.
In the sequel, unless otherwise specified, by \emph{learning graph} we mean \emph{extended learning graph}.

When there is no ambiguity, we usually define $\cS$ by stating the \emph{label} of each vertex.
We also say that an edge %\red{**on parle parfois d'arc et parfois d'edge. Uniformiser?**} 
$e=(u,v)$ \emph{loads} $j$ when $\cS(v)=\cS(u)\cup\{j\}$.

%by only separating positive and negative weights.% and stages of the learning graph $\G$. By \emph{level} $d$ of $\G$ we refer to the set of vertices at distance $d$ from the root. A \emph{stage} is the set of edges of $\G$ between level $i$ and level $j$, for some $i < j$.
The complexity of  extended learning graphs is  defined similarly to the one of other learning graphs 
by choosing the appropriate weight function for each complexity terms.
\begin{definition}[Extended learning graph complexity]\label{extlg}
Let $\G$ be an extended learning graph for a function $f : \ZZ \rightarrow \B$. 
Let $x\in \ZZ\setminus\YY$, $y\in\YY$,   and $\F \subseteq \E$.
The \emph{negative complexity} of $\F$ on $x$ and the \emph{positive complexity} of $\F$ on $y$ 
(with respect to $\G$) are respectively defined by
$$C^0(\F,x)=\sum_{e \in \F} w_x^0(e) \quad\text{and}\quad
C^{1}(\F,y)= %\max_{y \in f^{-1}(1)} \left( 
\sum_{e \in \F}  \frac{p_y(e)^2}{w_y^1(e)}.$$
%\right).
%\end{equation*}
%The positive complexity of $E$ is $$C^1(E) =\max_{y \in Y} C_{1,y}(E).$$
Then the \emph{negative and positive complexities} of $\F$ are 
$C^0(\F) =\max_{x\in f^{-1}(0)}C^0(\F,x)$ and $C^1(\F)=\max_{y\in f^{-1}(1)}C^1(\F,y)$.
The \emph{complexity} of $\F$ is $C(\F)=\sqrt{C^0(\F)C^1(\F)}$ and
the  \emph{complexity} of $\G$ is $C(\G)= C(\E)$.  
Last, the \emph{extended learning graph complexity} of  $f$, denoted $\ELG(f)$, is the minimum complexity of 
an extended learning graph for $f$.
%	Let $\G$ be a learning graph, $z \in \ZZ$ , $e \in \EE$ and $E \subseteq \EE$ .
%	The negative and positive complexities of $e$ on $x$ (resp. of $E$ on $x$) are respectively equal to the negative weight and the inverse of the positive weight (if $w^1_e(z)>0$) -- that is
%	\begin{align*}
%		C^0(e,x) = w^0_e(x), \quad C^1(e,x) = \frac{1}{w^1_e(z)}
%		\qquad
%		\left( \hbox{resp.} \quad C^0(E,z) = \sum_{e \in E} C^0(e,z), \quad C^1(e,z) = \sum_{e \in E} C^1(e,z) \right).
%	\end{align*}
%	The positive and negative complexities of $e$ are defined as $C^0(e) = \max_{z\in \ZZ \backslash \YY} C^0(e,z)$ and $C^1(e) = \max_{z\in \YY} C^1(e,z)$. We define similarly the positive and negative complexities of $E$, denoted respectively $C^0(E)$ and $C^1(E)$.
%	
%	The complexity of $e$ is $C(e) = \sqrt{C^0(e) C^1(e)}$, and the complexity of $E$ is $C(E) = \sqrt{C^0(E) C^1(E)}$.
%	
%	The complexities of the learning graph $\G$ are the complexities of its set of edges -- that is,
%	\begin{align*}
%		& C^0(\G,z) = C^0(\EE,z), C^1(\G,z) = C^1(\EE,z), \quad \forall z \in \ZZ,
%		\\
%		& C^0(\G) = C^0(\EE), C^1(\G) = C^1(\EE), \hbox{ and } C(\G) = C(\EE).
%	\end{align*}
%	
%	The extended learning graph complexity of a function $f$, denoted $\ELG(f)$ is the minimum complexity of an extended learning graph for $f$.
\end{definition}
Most often we will split a learning graph into \emph{stages} $\F$, that is,
when the flow through $\F$ has the same total amount $1$ for every positive inputs.
This allows us to analyze the learning graph separately on each stage.

As for adaptive learning graphs~\cite{spanCert,BL11}, the extended learning graph complexity is upper bounding the standard query complexity.
\begin{theorem}
\label{thm:ELG-querycomplexity}
%	Assume $\G$ is a learning graph for a function $f : \mathcal D \to \{0,1\}$ with $\mathcal D \subseteq \{0,1\}^n$. Then there exists a bounded-error quantum query algorithm for the same function with complexity $O(C(\G, f))$. So $Q(f) = $O(\ELG(f))$.
For every function $f : \ZZ \rightarrow \B$, we have
$Q(f)=O(\ELG(f))$.  
\end{theorem}
\begin{proof}
%(see the proofs of Theorems 5 and 7 in Belovs' paper for $k$-disctinctness)
We assume that $f$ is not constant, otherwise the result holds readily. The proof follows the lines of the analysis of the learning graph for Graph collision in~\cite{bel12}.
We already know that $Q(f)=O(\Advpm(f))$ by Theorem~\ref{qadv}.
Fix any extended learning graph $\G$ for $f$.
Observe from Definition~\ref{def:advpm-bound-sdp} that $\Advpm(f)$ is defined by a minimization problem.
Therefore finding any feasible solution  with objective value $C(\G, f)$ would conclude the proof.
Without loss of generality, assume that $C^0(\G)=C^1(\G )$ (otherwise we can
multiply all weights by $\sqrt{C^1(\G )/C^0(\G )}$).
Then both complexities become $\sqrt{C^0(\G )C^1(\G )}$ and the total complexity remains $C(\G)$.

	For each edge $e=(u,v)\in \E$ with $\cS(v)=\cS(u) \cup \{j\}$, define a block-diagonal matrix $X_j^e = \sum_\alpha (Y_j^e)_\alpha$, where the sum is over all possible assignments $\alpha$ on $\cS(u)$. Each $(Y_j^e)_\alpha$ is defined as $(\psi_0\psi_0^* + \psi_1\psi_1^*)$, where  for each $z \in \{0,1\}^n$ and $b\in\B$
$$
	\psi_b[z] =
	\begin{cases}
		p_e(z) / \sqrt{w_z^1(e)} &\hbox{if $z_{\cS(u)}=\alpha$, $f(z) =1$ and $ z_j = 1-b$,}
%		p_e(x) / \sqrt{w_e^1(z_S,z_j)} &\hbox{if $z$ satisfies $\alpha$ and $f(z) = z_j = 1$}
		\\
		\sqrt{w_z^0(e)} &\hbox{if $z_{\cS(u)}=\alpha$, $f(z) =0$ and $z_j = b$,}
		\\
		0  & \hbox{otherwise.}
	\end{cases}$$
%	\\
%	\hbox{and}\quad &\phi[z] =
%	\begin{cases}
%		p_e(z) / \sqrt{w_e^1\left(z_S,z_j\right)} &\hbox{if $z$ satisfies $\alpha$ and $f(z)=1, z_j = 0$}\\
%%		p_e(x) / \sqrt{w_e^0\left(z_S,z_j\right)} &\hbox{if $z$ satisfies $\alpha$ and $f(z)=1, z_j = 0$}
%		\\
%		\sqrt{w_e^0\left(z_S,z_j\right)} &\hbox{if $z$ satisfies $\alpha$ and $f(z)=0, z_j = 1$}
%		\\
%		0  & \hbox{otherwise}
%	\end{cases}
%($\psi$ and $\phi$ depend implicitly on $j, e$ and $\alpha$).

Define now $X_j = \sum_e X_j^e$ where the sum is over all edges $e$ loading $j$. Fix any $x \in f^{-1}(0)$ and $y \in f^{-1}(1)$. Then we have
$X_j^e[x,x] = 		w_x^0(e)$ and
$X_j^e[y,y] = (p_e(y))^2 / w_y^1(e)$.
So the objective value is
\begin{align*}
	\max_{z \in \{0,1\}^n} \sum_{j\in[n]} X_j [z,z]
	&= \max\left\{ \max_{x \in f^{-1}(0)} \sum_{j} X_j [x,x] , \max_{y \in f^{-1}(1)} \sum_{j} X_j [y,y] \right\} \\
	&= \max\left\{ C^0(\G ) ,  C^1(\G ) \right\} = C(\G). %\\
%	&= C(\G), \quad\text{since we have assumed that $C^0(\G ) =C^1(\G ) $.}
\end{align*}

Fix now $x \in f^{-1}(0)$ and $y \in f^{-1}(1)$. 
Consider the cut $\F$  over $\G$ of edges $(u,v)\in\E$
such that $\cS(v)=\cS(u) \cup \{j\}$ and $x_{\cS(u)} = y_{\cS(u)}$ but $x_j \neq y_j$. 
Then each edge $e\in\F$ loading $j$ satisfies $w_x^0(e) = w_y^1(e)$
and therefore
$X_j^e[x, y]  =
		p_e(y)$. % &\hbox{ if } x_S = y_S, x_j \neq y_j 
%\end{align*}
Thus, 
$
	\sum_{{j :\, x_j \neq y_j}} \; X_j [x,y] = \sum_{e\in \F} p_e(y)=1
$.
Hence  constraints of Definition~\ref{def:advpm-bound-sdp} are satisfied.
\end{proof}

%\section{Super edges}
% subsubsubsubsubsubsubsubsubsubsubsubsubsubsubsubsubsubsubsub %
\subsection{Compression of learning graphs into super edges}

We will simplify the presentation of our learning graphs by introducing a new type of edge
encoding specific learning graphs as sub-procedures.
%The terminology is justified by the fact that a super edge can be embedded in a larger LG as a normal edge, 
Since an edge has a single `exit', we can only encode learning graphs whose flows have unique sinks.
%Since any flow going through such between two nodes (the amount of flow that enters the super edge is the same as the amount that exits through its sink). Furthermore, a (normal) edge is a particular type of super edge: it is a LG reduced to two nodes and a single edge.

\begin{definition}[Super edge]\label{def:super}
	A \emph{super edge} $e$ is an extended learning graph $\G_e$ such that each possible flow has the same unique sink.
%Since a super edge is just a particular type of LG, its complexities are the complexities of the underlying LG: given a super edge $\mathbf e$, 
Then its \emph{positive} and \emph{negative edge-complexities} on input $x\in\ZZ\setminus\YY$ and $y\in\YY$ are respectively $c^0( e, x)=C^0(\G_e, x)$ and $c^1( e, y)=C^1(\G_e, y)$.
\end{definition}
Since an edge is also a super edge with unit flow, we will also use the notion of positive and negative edge-complexities for edges for simplifying further notations.
Observe that in particular $C^0(e,x)=c^0(e,x)$ and $C^1(e,y)=p_y(e)^2\times c^1(e,y)$,
which we now use in order to define the complexity of learning graphs with super edges.

Indeed, one can now consider learning graphs with super edges.
They are equivalent to learning graphs without super edges by doing recursively
the following replacement for each super edge:
(1) replace it by its underlying learning graph, plugging the root to
all incoming edges and the unique flow sink to all outgoing edges; 
(2) root the incoming flow accordingly to the plugged learning graph.
%
%This observation is the basis case for defining the complexity of a learning graphs with super edges.
Let us call this learning graph the \emph{expansion} of the original one with super edges.
Then, a direct inspection leads to the following result that we will use in order to compute complexities directly on our original learning graphs.
\begin{lemma}\label{expansion}
Let $\G$ be a learning graph with super edges for some function $f$. 
Then the expansion of $\G$ is also a learning graph for $f$.
Moreover, let $\exp(\F)$ be the expansion of $\F\subseteq\E$.
Then $\exp(\F)$ has positive and negative complexities 
$$C^0(\exp(\F),x)=\sum_{e \in \F} c^0(e,x) \quad\text{and}\quad
C^1(\exp(\F),y)= %\max_{y \in f^{-1}(1)} \left( 
\sum_{e \in \F}  p_y(e)^2\times c^1(e,y).$$
\end{lemma}
%When it is clear from the context
%Since an edge is also a super edge, we will also use the notion of positive and complexities for edges without any 

%\begin{definition}[complexity of a super edge]
%	Given a super edge $\mathbf e$, its \emph{positive} and \emph{negative complexities} on input $z \in \{0,1\}^n$ are respectively $C^0(\mathbf e, z)$ and $C^1(\mathbf e, z)$, the complexities of the underlying learning graph $\G$.
%\end{definition}

Fix some stage $\F\subseteq\E$ of $\G$, that is such that
the flow through $\F$ has the same total amount $1$ for every positive inputs.
We will use the following lemma, that we adapt from non-adaptive learning graphs, to assume that a learning graph has positive complexity at most $1$ on $\F$. The parameter $T$ involved in this result is usually called the \emph{speciality} of $\F$. 
%Most often $\F$ will correspond to a stage, and therefore its total amount of flow will be $1$.
\begin{lemma}[Speciality~\cite{bel12}]\label{lem:specExtended}
Let $\G$ be a learning graph for a function $f : \ZZ \rightarrow \B$.
%Consider a learning graph $\G = (\VV, \EE, w, C, \{p_y\,:\, y\in\YY\})$. Denote $C^0$ and $C^1$ its negative and the positive edge complexities.
Let $\F\subseteq \E$ be a stage of $\G$ %such that
%Let $\F$ be a stage of $\G$ such that:
whose flow % for every positive inputs%through $\F$ %has total amount $1$,
always uses the same ratio $1/T$ of transitions and is uniform on them. %those transitions.
Then there are alternative weights for edges in $\F$ such that the
new weighted edges in $\widetilde{\F}$ satisfy
for every $x\in f^{-1}(0)$ and $y\in f^{-1}(1)$ %the new complexities satisfy
%$\G$ which coincides with $w^b$ outside $\F$
%such that, denoting $\tilde C^0$ and $\tilde C^1$ the negative and the positive edge complexities in $\tilde \G = (\VV, \EE, \tilde w, \{p_y\,:\, y\in\YY\})$,
%Let $C^0$ and $C^1$ be the negative and the positive complexities of the edges in $\F$. They can be rescaled such that the complexities of $\F$ become 
$${C}^0(\widetilde\F, x) \leq T \Exp_{e\in\F}\left[ c^0(e,x) c^1(e) \right] \quad\text{and}\quad %\text{for every $x\in f^{-1}(0)$,}\\
	%\text{and}\quad 
	{C}^1(\widetilde\F, y) \leq 1. %, \qquad \text{for every $y\in f^{-1}(1)$.}
$$
\end{lemma}
%We can also get Lemma~\ref{lem:specExtended-balanced}, with balanced complexities.
\begin{proof}
Let $n_\mathrm{total}$ be the number of transitions in $\F$ and $n_\mathrm{used}$ the number of them used by each flow (i.e. with positive flow).
Therefore $T= n_\mathrm{total}/n_\mathrm{used}$. 
By assumption, the flow on each edge is either $0$ or $1/n_\mathrm{used}$.
%, for some $q>0$.
For each edge $e$ in $\F$, let $\lambda_e = {c^1(e)}/ n_\mathrm{used}$. 
%where we remind that $C^1(e)= \max_{y\in f^{-1}(1)} p_y(e)^2/w^1_y(e)$. %C^1(e,y)$. 
For every input $z$, we multiply $w^b_z(e)$ by $\lambda_e$, % w^b_z(e)$.
%and  $w^1_z(e)$ to $\lambda_e w^0_z(z)$
and we name $\widetilde\F$ the set $\F$ with the new weights.

Then for any $x \in f^{-1}(0)$,
$
	 C^0(\widetilde\F,x) = \sum_{e\in\widetilde\F}  \lambda_e c^0(e,x) = T \Exp_{e\in\widetilde\F}\left[ c^0(e,x) c^1(e) \right]$.
Similarly, for any $y \in f^{-1}(1)$,
%$\tilde C^1(\widetilde\F,y) = \sum_e \frac{1}{ \lambda_e} \frac{1}{(n_\mathrm{used})^2} C^1(e,y) \leq 1$ 
$C^1(\widetilde\F,y) = \sum_{e\in\widetilde\F} {p_y(e)^2} c^1(e,y) / { \lambda_e}
\leq %\sum_{e\in\widetilde\F} p_y(e)^2 n_\mathrm{used} =
 1
$,
since terms in the sum are positive only for edges with positive flow.
\end{proof}

\subsection{Loading sparse inputs}

We study a particular type of super edges, that we will use repeatedly in the sequel.
%\begin{definition}[LEARN]
%	IDEA : given a set $S$, $LEARN(S,S \cup S')$ : load the elements in $S$.
%\end{definition}
	To construct a learning graph for a given function, one often needs to load a subset $S$ of the labels. This can be done by a path of length $|S|$ with negative and positive complexities $|S|$,
	which, after some rebalancing, leads directly to the following lemma.
%	or after rebalancing nega
\begin{lemma}\label{lem:complexityLearnAll}
	For any set $S$, there exists a super edge denoted $\Loadall_S$ loading $S$ with the following complexities for any input $z \in \ZO^N$:
		$$
		c^0(\Loadall_S, z)  = |S|^2  \quad\text{and}\quad
		c^1(\Loadall_S, z) = 1 .
		$$
%	where $|S|$ denote the cardinal of $S$. So the total complexity is $C(\Loadall_S, z) = |S|$.
\end{lemma}
\iffalse\begin{proof}
	Let us assume for simplicity that $N = |S|$ and denote $S = \{1, \dots, N\}$. We define the learning graph $\Loadall_S$ as the sequence of edges $e_1 = (\emptyset, \{1\}), e_2 = (\{1\},\{1,2\}), e_N = ([N-1],|S|)$. The weights are defined as:
	$$
		w_{e_k}^0 = N
		\quad
		\hbox{ and }
		\quad
		w_{e_k}^1 = N.
	$$
	They satisfy the requirements of Definition~\ref{def:extLGfct}.
	We obtain the values of the complexities by using their definition.
\end{proof}
\fi

%	The goal is then to design an efficient super edge loading $S$.	 Let us first point out that there is a straightforward (non-adaptive) learning graph to load the elements of $S$.

When the input is sparse %(i.e., has a low Hamming weight) on $|S|$, 
one can do significantly better as we describe now, where $|z_S|$ denotes the Hamming weight of $z_S$.
\begin{lemma}\label{lem:complexityLearn}
	For any set $S$, there exists a super edge denoted $\Loadsparse_S$ loading $S$ with the following complexities for any input $z \in \ZO^N$:
%	$$
%		C^0(\Loadsparse_S,z)  \leq
%		O(|S| \cdot \log(|S|) \cdot (|z_S|+1)),
%%		\begin{cases}
%%			O(|S| \cdot \log(|S|)) &\hbox{ if } |z_S| = 0 \\
%%			O(|S| \cdot (|z_S|+1) \cdot \log(|S|)) &\hbox{ otherwise,}
%%		\end{cases}
%		\qquad
%		C^1(\Loadsparse_S,z) \leq O(1)
%	$$
%	{\color{red} OR RATHER (to deal with the case where $|S|=1$, hence $\log|S|=0$):
		$$
		c^0(\Loadsparse_S,z)  \leq
		6 |S|  (|z_S|+1) \log(|S|+1) \quad\text{and}\quad
%		\begin{cases}
%			O(|S| \cdot \log(|S|)) &\hbox{ if } |z_S| = 0 \\
%			O(|S| \cdot (|z_S|+1) \cdot \log(|S|)) &\hbox{ otherwise,}
%		\end{cases}
		\quad
		c^1(\Loadsparse_S,z) \leq 1.
	$$
%	}
\end{lemma}
\begin{proof}
	Let us assume for simplicity that $N = |S|$ and  $S = \{1, \dots, N\}$. We define the learning graph $\Loadsparse_S$ as the path through edges $e_1 = (\emptyset, \{1\})$,  $e_2 = (\{1\},\{1,2\})$, $\dots$, $e_N = (\{1,\dots,N-1\},S)$. The weights are defined as, for $b\in\{0,1\}$ and $z\in\ZZ$,
$$
		w_{e_k}^b(z) = 
		\begin{cases}
			3\cdot  (|z_{[j-1]}| +1) \cdot \log(N+1) &\hbox{if } z_j = b, \\
			3 N \cdot \log(N+1) &\hbox{if } z_j = 1-b,
		\end{cases}
	$$
\iffalse $$
		w_{e_k}^0(z) = 
		\begin{cases}
			3\cdot  (|z_{[j-1]}| +1) \cdot \log(N+1) &\hbox{if } z_j = 0 \\
			3 N \cdot \log(N+1) &\hbox{if } z_j = 1,
		\end{cases}
	$$
	and
	$$
		w_{e_k}^1(z) = 
		\begin{cases}
			3 N \cdot \log(N+1) &\hbox{if } z_j = 0 \\
			3\cdot  (|z_{[j-1]}| +1) \cdot \log(N+1) &\hbox{if } z_j = 1.
		\end{cases}
	$$\fi	
	When $|z|>0$, let us denote $i_0 = 0$, $i_{|z|+1}=N+1$ and $(i_k)_{k=1, \dots, |z|}$ the increasing sequence of indices $j$ such that $z_j = 1$. 
	%, that is, for all $1 \leq k \leq |z|$, $z_{i_k} = 1$ and $i_{k-1} < i_{k}$. 
	Then, for $k=1, \dots, |z|+1$, we define $m_k$ as the number of indices $j\in(i_{k-1},i_k)$ such that $z_j = 0$. More precisely, $m_k = i_{k} - i_{k-1}-1$ for $1 \leq k \leq |z|$ and $m_{|z|+1} = N - i_{|z|}$. So $\sum_{k = 1}^{|z|+1} m_k = N - |z|$.
	Then, for any input $z$,
	\begin{align*}
		C^0(\Loadsparse_S, z)
		&=
		\begin{cases}
			3N \cdot \log(N+1) & \text{ if } |z| = 0,\\
			3\cdot  \Big( |z| N + \sum_{i=1}^{|z|+1} i \times m_i \Big) \cdot \log(N+1) & \hbox{otherwise,}
		\end{cases}
%		\\
%		&\leq
%		6N  \cdot (|z|+1) \cdot \log(N+1).
	\end{align*}
	which is bounded above by $6N  \cdot (|z|+1) \cdot \log(N+1)$.
	Moreover, using $\sum_{i=1}^{|z|+1} \frac{1}{i} \leq \log(|z|+1)+1$, we get
	\begin{align*}
		C^1(\Loadsparse_S, z)
		=
			\frac{1}{3 \cdot \log(N+1)} \left((N-|z|)\frac{1}{N} + \sum_{i=1}^{|z|+1} \frac{1}{i} \right)
		\leq 1.
	\end{align*}
\end{proof}
%We can also get balanced complexities, see Lemma~\ref{lem:complexityLearn-balanced}.

% subsubsubsubsubsubsubsubsubsubsubsubsubsubsubsubsubsubsubsub %
%\subsection{Notations}

% subsubsubsubsubsubsubsubsubsubsubsubsubsubsubsubsubsubsubsub %
\section{Composition of learning graphs}
\label{subsec:decomp-lemmas}

To simplify our presentation, we will use the term \emph{empty transition} for an edge between two vertices representing the same set. 
They carry zero flow and weight, and they do not contribute to any complexity.
% them to obtain an actual learning graph.

%\label{subsec:graphical-notations-LG}

%We introduce two learning graphs and some graphical notations to simplify the presentation.
%The complexities of these learning graphs are analyzed in~\S~\ref{subsec:decomp-lemmas}.

% PPPPPPPPPPPPPPPPPPPPPPPPPP %
\subsection{Learning graph for OR}
Consider $n$ Boolean functions $f_1,\ldots,f_n$ with respective learning graphs $\G{}_1,\ldots,\G{}_n$.
The following lemma explains how to design a learning graph $\GOR$ for $f=\bigvee_{i\in[n]} f_{i}$
whose complexity is the squared mean of former ones.
We will represent $\GOR$ graphically as\\
\centerline{\begin{tikzpicture}
\node (A) at (0,0){$\emptyset$};
\node (B) at (2,0){$\G{}_{i}$};
\draw[->,>=latex] (A) -- (B) node[above,midway,sloped] (C){$i$};
\end{tikzpicture}}
This result is similar to the one of~\cite{amb10}, where
a search procedure is designed for the case of variable query costs,
or equivalently for a search problem divided into subproblems with variable complexities.

%This result relies on the learning graph $\GOR$ presented in~\S~\ref{subsec:graphical-notations-LG}.
\begin{lemma}\label{groverExtended}
Let $\G_1,\ldots,\G_n$ be learning graphs for Boolean functions $f_{1},\ldots,f_{n}$ over $\ZZ$.
%where each $f_i$ can be computed by a learning graph $\G_i$.
Assume further that for every $x$ such that $f(x)=1$, there is at least $k$ functions $f_{i}$ such that $f_{i}(x)=1$.
Then there is a learning graph $\G$ for $f=\bigvee_{i\in[n]} f_{i}$ such that
for every  $z\in\ZZ$
$$
	\begin{cases}
		\displaystyle C^0(\G, z)\leq {\frac{n}{k}} \times \Exp_{i\in[n]}  \left(C^0(\G_i,z) C^1(\G_i)\right)& \text{when $f(z)=0$},\\
		\displaystyle C^1(\G, z)\leq 1&  \text{when $f(z)=1$}.
	\end{cases}
$$
\end{lemma}
\begin{proof}
We define the new learning graph $\G$ by considering a new root $\emptyset$ that we link to the roots of each $\G{}_i$. In particular, each $\G{}_{i}$ lies in a different connected component. 
For $n=3$, the graph is displayed below:\\ %[-2ex]
\centerline{\begin{tikzpicture}
\node (A) at (0,0){$\emptyset$};
\node (B) at (-2,-1){$\G{}_{1}$};
\node (C) at (0,-1){$\G{}_{2}$};
\node (D) at (2,-1){$\G{}_{3}$};
\draw[->,>=latex] (A) -- (B);
\draw[->,>=latex] (A) -- (C);
\draw[->,>=latex] (A) -- (D);
\end{tikzpicture}}
%The rest of the proof consist in proving that $\G{}$ has learning graph complexity $C=O\left(\sqrt{\frac{\sum_{i}(C^{(i)})^{2}}{k}}\right)$. 
Then, we rescale the original weights of edges in each component $\G{}_{i}$ 
 by  $\lambda_i=C^1(\G_i)/k$.
%The main learning graph $\GOR$ is decomposed into a basic search problem where each check is performed by a sub-learning graph.%Let $\G{}_{i}$ be  learning graphs  for functions $f_i$  with complexities $C_{i}(x)$ on input $x$, and 
%Construct the learning graph $\G$ as follows.
%{\color{red}
%Consider a new root $\emptyset$ that we link to the roots of each $\G{}_i$. In particular, each $\G{}_{i}$ leads to different connected components.
%\begin{center}
%\begin{tikzpicture}
%\node (A) at (0,0){$\emptyset$};
%\node (B) at (-2,-2){$\G{}_{1}$};
%\node (C) at (0,-2){$\G{}_{2}$};
%\node (D) at (2,-2){$\G{}_{3}$};
%\draw[->,>=latex] (A) -- (B);
%\draw[->,>=latex] (A) -- (C);
%\draw[->,>=latex] (A) -- (D);
%\end{tikzpicture}
%\end{center}
%%The rest of the proof consist in proving that $\G{}$ has learning graph complexity $C=O\left(\sqrt{\frac{\sum_{i}(C^{(i)})^{2}}{k}}\right)$. 
% Then, we rescale the original weights of edges in each component $\G{}_{i}$ 
% by a factor of $\lambda_i=C^1(\G_i, f_i)/k$.}

The complexity $C^0(\G, x)$ for a negative instance $x$ is
$$	C^0(\G,x)= \sum_{i=1}^{n}  \lambda_i C^0(\G_i,x) = {\frac{n}{k}}\times\Exp_i  \left(C^0(\G_i,x) C^1(\G_i)\right).$$

Consider now a positive instance $y$.
Then $y$ is also a positive instance for at least $k$ functions $f_i$.
Without loss of generality assume further that these $k$ functions are $f_1,f_2,\ldots,f_k$.
We define the flow of $\G$ (for $y$) as a flow uniformly directed from $\emptyset$ to $\G{}_{i}$ for $i=1,2,\ldots,k$. In each component $\G{}_{i}$, the flow is then routed as in $\G{}_{i}$.
Therefore we have
$$	C^1(\G, y)=\sum_{i=1}^{k} \frac{1}{k^2}\times \frac{C^1(\G_i,y)}{\lambda_i } \leq 1.$$
Finally, observe that by construction the flow is directed to sinks having $1$-certificates, thus
$\GOR$ indeed computes $f=\bigvee_{i\in[n]} f_{i}$.
\end{proof}

% PPPPPPPPPPPPPPPPPPPPPPPPPP %
\subsection{Learning graph for Johnson walks}
We build a framework close to the one of quantum walk based algorithms from~\cite{mss07,mnrs11} but for extended learning graphs.
To avoid confusion we encode into a partial assignment the corresponding assigned location, that is,  $z_{S}=\{(i,z_i):i\in S\}$.

Fix some parameters $r\leq k\leq n$.
% consider Boolean functions $f_{A}$ over $\ZZ$.
We would like to define a learning graph $\GJohnson$ for $f=\bigvee_{A} f_{A}$, 
where $A$ ranges over $k$-subsets of $[n]$ and $f_A$ are Boolean functions over $\ZZ$, but differently than in Lemma~\ref{groverExtended}.
%For every subset $A'$ of $[n]$,  let $I(A')\subseteq [N]$ be some index set depending on $A$ only.
For this, we are going to use a learning graph for $f_A$ when the input has been already partially loaded, that is, loaded on $I(A)$ for some subset $I(A)\subseteq [N]$ depending on $A$ only.
Namely, we assume we are given, for every partial assignment $\lambda$,
%each input $z$, 
a learning graph $\G_{A,\lambda}$ defined over inputs $\ZZ_\lambda=\{z\in\ZZ : z(I(A))=\lambda\}$ for $f_A$ restricted to $\ZZ_\lambda$. 

Then, instead of the learning graph of Lemma~\ref{groverExtended},
our learning graph $\GJohnson$ factorizes the load of input $z$ over $I(A)$ for $|A|=k$ and then uses $\G_{A,z_{I(A)}}$. This approach is more efficient
when, for every positive instance $y$, there is a $1$-certificate $I(T_y)$  for some $r$-subset $T_y$, and $A\mapsto I(A)$ is monotone.
This is indeed the analogue of a walk on the Johnson Graph.
%there exists an $r$-subset $T_x\subseteq [n]$ such that 
%$f_A(z)=1$ when $T_z\subseteq A$ and $z_A=y_A$.

We will represent the resulting learning graph $\GJohnson$ graphically using $r+1$ arrows: 
one for the first load of $(k-r)$ elements,  and $r$ smaller ones for each of the last $r$ loads of a single element. For example, when $r=2$ we draw:\\
\centerline
{\begin{tikzpicture}
\node (A) at (0,0){$\emptyset$}; %{};
\node (B) at (3,0){};
\node (C) at (3.5,0){};
\node (D) at (4.5,0){\;$\G_{A,x_{I(A)}}$};
\draw[-,>=latex] (A) -- (D) node[above,midway,sloped] (1){$A$};
\draw[->,>=latex] (A) -- (B);
\draw[->,>=latex] (B) -- (C);
\draw[->,>=latex] (C) -- (D);
\end{tikzpicture}}

%We analyze the complexity of the learning graph for Johnson walk $\GJohnson$ introduced in~\S~\ref{subsec:graphical-notations-LG}.
In the following, $\Learn_{S}$ denotes any super edge loading the elements of $S$,
such as $\Loadall$ or $\Loadsparse$ that we have defined in Lemmas~\ref{lem:complexityLearnAll} and~\ref{lem:complexityLearn}.
\begin{theorem}\label{johnsonExtended}
For every subset $S\subseteq[N]$, let $\Learn_{S}$ be any super edge loading $S$
with $c^1(\Learn_{S})\leq 1$.
Let $r\leq k\leq n$ and let $f=\bigvee_{A} f_{A}$, 
where $A$ ranges over $k$-subsets of $[n]$ and $f_A$ are Boolean functions over $\ZZ$. 

Let $I$ be a monotone mapping from subsets of $[n]$ to subsets of $[N]$
with the property that,  for every $y\in f^{-1}(1)$, 
there is an $r$-subset $T_y\subseteq [n]$ whose image  $I(T_y)$ is a $1$-certificate for $y$. % when $A\supseteq T_y$.

Let $\mathrm{\bf{S}},\mathrm{\bf{U}}>0$ be  such that  every $x\in f^{-1}(0)$ satisfies
\begin{eqnarray}
	 \label{j2}
\Exp_{A' \subset [n] \,:\,  |A'| = k-r}  \left(C^0(\Learn_{I(A')},z)\right) &\leq& \mathrm{\bf{S}}^2 \,;\\
	 \label{j3}
	\Exp_{A'\subset A''\subseteq [n] \,:\,  |A'| =|A''|-1= i}  \left(C^0(\Learn_{I(A'')\backslash I(A')},z) \right) &\leq&\mathrm{\bf{U}}^2,  \quad\text{for  $k-r\leq i <k$} \,.
\end{eqnarray}

Let $\G_{A,\lambda}$ be learning graphs for functions $f_A$ on $\ZZ$ 
 restricted to inputs $\ZZ_\lambda=\{z\in\ZZ : z(I(A))=\lambda\}$, for all $k$-subsets $A$ of $[n]$ and all possible assignments $\lambda$ over $I(A)$.
Let finally $\mathrm{\bf{C}}>0$ be such that every $x\in f^{-1}(0)$ satisfies
\begin{eqnarray}
		\Exp_{A\subseteq[n]\,:\,|A|=k} \left(C^0(\G_{A,x_{I(A)}},x)C^1(\G_{A,x_{I(A)}}, f)\right) &\leq &\mathrm{\bf{C}}^2.	\label{j4}
\end{eqnarray}

Then there is a learning graph $\GJohnson$ for $f$ such that for every  $z\in\ZZ$
$$
\begin{cases}
\displaystyle C^0(\GJohnson, z)= O\left(\mathrm{\bf{S}}^2+\left(\frac{n}{k}\right)^{r}\left(k\times \mathrm{\bf{U}}^2+ \mathrm{\bf{C}}^2\right)\right)&
\text{when $f(z)=0$}, \medskip\\
C^1(\GJohnson, z) = 1&
\text{when $f(z)=1$}.
\end{cases}
$$
%where $A$ is taken over $k$-subsets of $[n]$. % and $T_x=\emptyset$ for negative instances $x$.
\end{theorem}
%We can also obtain Lemma~\ref{johnsonExtended-C1not1}, where $C^1$ is not bounded by $1$.
\begin{proof}
\textbf{Construction.}
 We define $\GJohnson$ by emulating a walk on the Johnson graph $J(n,k)$ for searching  a $k$-subset $A$ having an $r$-subset $T_y$ such that
 $I(T_y)$ is a $1$-certificate for $y$.
 In that case, by monotonicity of $I$, the set $I(A)$ will be also a $1$-certificate for $y$.
  %=\{j_1,\ldots,j_r\}$ be a $1$-certificate of $y$.

Our learning graph $\GJohnson$ is composed of $(r+2)$ stages (that is, layers whose total incoming flow is $1$), that we call Stage~$\ell$, for $\ell=0,1,\ldots,r+1$.
%First we need to define an intermediate set of indices $I(B)$ for $(k-r)$-subsets $B$ of $[n]$ as 
%$I(B)=\bigcap_{A} I(A)$, where the intersection is taken over $k$-subsets $A$ of $[n]$ containing $B$.
%By hypothesis we know that $|I(B)|\leq \mathrm{\bf{S}}$ for all $(k-r)$-subsets $B$,
%and that $|I(A)\setminus I(B)|\leq \mathrm{\bf{U}}$ for all $k$-subsets $A$ containing $B$.
%
%All stages but the last one are non-adaptive.
%In that stages, all edges have constant weight. Scaling the weights of edges make positive and negative complexities equal to the global complexity of the stage.
%
An example of such a learning graph for $n=4$, $k=3$ and $r=1$ is represented below:\\
\centerline{%\textbf{Adapter cette figure pour $r=2$}
					\begin{tikzpicture}[scale=0.9]\scriptsize
					\node[draw,ellipse] (A) at (0,-2.7) {$\emptyset$};
%					\node[draw,circle,minimum height=0.8cm] (L) at (-2.5,-2) {$\{1\}$};
%					\node[draw,circle,minimum height=0.8cm] (M) at (-1.5,-2) {$\{1\}$};
%					\node[draw,circle,minimum height=0.8cm] (N) at (-0.5,-2) {$\{1\}$};
%					\node[draw,circle,minimum height=0.8cm] (O) at (0.5,-2) {$\{2\}$};
%					\node[draw,circle,minimum height=0.8cm] (P) at (1.5,-2) {$\{2\}$};
%					\node[draw,circle,minimum height=0.8cm] (Q) at (2.5,-2) {$\{3\}$};
%					
%					\draw[->,>=latex] (A) -- (L);
%					\draw[->,>=latex] (A) -- (M);
%					\draw[->,>=latex] (A) -- (N);
%					\draw[->,>=latex] (A) -- (O);
%					\draw[->,>=latex] (A) -- (P);
%					\draw[->,>=latex] (A) -- (Q);
%					
					\node[draw,ellipse] (B) at (-3.5,-4) {$\{1,2\}$};
					\node[draw,ellipse] (C) at (-2.1,-4) {$\{1,3\}$};
					\node[draw,ellipse] (D) at (-0.7,-4) {$\{1,4\}$};
					\node[draw,ellipse] (E) at (0.7,-4) {$\{2,3\}$};
					\node[draw,ellipse] (F) at (2.1,-4) {$\{2,4\}$};
					\node[draw,ellipse] (G) at (3.5,-4) {$\{3,4\}$};
%					
%					\draw[->,>=latex] (L) -- (B);
%					\draw[->,>=latex] (M) -- (C);
%					\draw[->,>=latex] (N) -- (D);
%					
%					\draw[->,>=latex] (O) -- (E);
%					\draw[->,>=latex] (P) -- (F);
%					\draw[->,>=latex] (Q) -- (G);
					\draw[->,>=latex] (A) -- (B);
					\draw[->,>=latex] (A) -- (C);
					\draw[->,>=latex] (A) -- (D);
					\draw[->,>=latex] (A) -- (E);
					\draw[->,>=latex] (A) -- (F);
					\draw[->,>=latex] (A) -- (G);					
					\node[draw,ellipse] (H) at (-3,-5.5) {$\{1,2,3\}$};
					\node[draw,ellipse] (I) at (-1,-5.5) {$\{1,2,4\}$};
					\node[draw,ellipse] (J) at (1,-5.5) {$\{1,3,4\}$};
					\node[draw,ellipse] (K) at (3,-5.5) {$\{2,3,4\}$};
					\draw[->,>=latex] (B) -- (H);
					\draw[->,>=latex] (C) -- (H);
					\draw[->,>=latex] (E) -- (H);
					\draw[->,>=latex] (B) -- (I);
					\draw[->,>=latex] (D) -- (I);
					\draw[->,>=latex] (F) -- (I);
					\draw[->,>=latex] (C) -- (J);
					\draw[->,>=latex] (D) -- (J);
					\draw[->,>=latex] (G) -- (J);
					\draw[->,>=latex] (E) -- (K);
					\draw[->,>=latex] (F) -- (K);
					\draw[->,>=latex] (G) -- (K);
					\end{tikzpicture}}
%					\end{center}

Stage $0$ of $\GJohnson$ consists in $n \choose (k-r)$ disjoint paths, all of same weights, leading to vertices labelled by some $(k-r)$-subset $A'$ and loading $I(A')$. They can be implemented by the super edges $\Learn_{I(A')}$.
For positive instances $y$, the flow goes from $\emptyset$ to subsets $I(A')$
such that $I(A')\cap T_y=\emptyset$.

For $\ell=1,\ldots,r$, Stage~$\ell$ consists in $(n-(k-r)-\ell+1)$ outgoing edges to each node labeled by a $(k-r+\ell-1)$-subset $A'$.
Those edges are labelled by $(A',j)$ where $j\not\in A'$ and load $I(A'\cup\{j\})\setminus I(A')$. 
%Their lengths are at most $\mathrm{\bf{U}}$ by Hypothesis~\ref{j3}.
They can be implemented by the super edges $\Learn_{I(A'\cup\{j\})\setminus I(A')}$. 
For positive instances $y$, for each vertex $A'$ getting some positive flow, the flow goes out only to the edge $(A',j_\ell)$, with the convention $T_y=\{j_1,\ldots,j_r\}$.

The final Stage~$(r+1)$ consists in plugging in nodes $A$
the corresponding learning graph $\G_{A,x_{I(A)}}$, for each $k$-subset $A$. 
%This is the only adaptive part of our global learning graph $\GJohnson$. 
We take a similar approach than in the construction of $\GOR$ above.
The weights of the edges in each component  $\G_{A,z_{I(A)}}$ are rescaled by a factor
$\lambda_A=C^1(\G_{A,x_{I(A)}})/{{n-r \choose k-r}}$. 
For a positive instance $y$, %Then
the flow is directed uniformly to each $\G_{A,y_{I(A)}}$ such that $T(y)\subseteq A$, and then according to $\G_{A,y_{I(A)}}$.

Observe that by construction, on positive inputs the flow reaches only $1$-certificates of $f$.
Therefore $\GJohnson$ indeed computes $f$.
\smallskip

\textbf{Analysis.}
Remind that the positive edge-complexity of our  super edge $\Learn$
%for loading any set $S$ 
is at most $1$. % $c^1(\Learn_{S})\leq 1$.

At Stage $0$, the $n \choose (k-r)$ disjoint paths are all of same weights.
The flow satisfies the hypotheses of Lemma~\ref{lem:specExtended} with a speciality of $O(1)$.
Therefore, using inequality~\eqref{j2}, the complexity of this stage is $O(\mathrm{\bf{S}}^2)$ when $f(x)=0$, and at most $1$ otherwise.

For $\ell=1,\ldots,r$, at Stage~$\ell$ consists of $(n-(k-r)-\ell+1)$ outgoing edges to each node labeled by a $(k-r+\ell-1)$-subset.
Take a positive instance $y$. %, and  set $T_x=\{j_1,\ldots,j_r\}$.
Recall that, for each vertex $A'$ getting some positive flow, the flow goes out only to the edge $(A',j_\ell)$.
By induction on $\ell$, the incoming flow is uniform when positive. Therefore, the flow on each edge with positive flow is also uniform, and the speciality of the stage is $O((\frac{n}{k})^\ell\cdot {k})$.
Hence, by Lemma~\ref{lem:specExtended} and using
 inequality~\ref{j3}, the cost of each such stage is $O((\frac{n}{k})^\ell \cdot k \cdot \mathrm{\bf{U}}^2)$.
The dominating term is thus $O\left((\frac{n}{k})^r\cdot k\cdot \mathrm{\bf{U}}^2 \right)$.

The analysis of the final stage (Stage~$(r+1)$) is similar to the proof of Lemma~\ref{groverExtended}. For a negative instance $x$, the complexity of this stage is:
\begin{align*}
	\sum_{A} \lambda_A{C^0(\G_{A,x_{I(A)}}, x) }
	&= {\frac{{n\choose k}}{{n-r \choose k-r}}} \Exp_{A} \left({C^0(\G_{A,x_{I(A)}}, x) C^1(\G_{A,x_{I(A)}})}\right)
	\\
	&= O\left( \left(\frac{n}{k}\right)^{r}  \times \Exp_{A} \left({C^0(\G_{A,x_{I(A)}}, x) C^1(\G_{A,x_{I(A)}})}\right)   \right).
\end{align*}
Similarly, when $f(y)=1$, we get a complexity at most $1$.
\end{proof}

%Taking $\Learn = \Loadsparse$ (see Lemma~\ref{lem:complexityLearn}) in Lemma~\ref{johnsonExtended}, we obtain Corollary~\ref{johnsonExtendedLearn}.

% SSSSSSSSSSSSSSSSSSSSSSSSSSSSSSSSSSSSSSSSSSSSSSSSSSSSSSSSSSSSSSSSSSSSSSSSSSSSSSSSSSSSSSSSSSSSSSSSSSSSSSSSSSSSSSSSSSSSSSSSSSSSSSSSSSSSSSSSSSSS %
\section{Application to Triangle Finding}\label{sec:triangle}
 % SSSSSSSSSSSSSSSSSSSSSSSSSSSSSSSSSSSSSSSSSSSSSSSSSSSSSSSSSSSSSSSSSSSSSSSSSSSSSSSSSSSSSSSSSSSSSSSSSSSSSSSSSSSSSSSSSSSSSSSSSSSSSSSSSSSSSSSSSSSS %

\subsection{An adaptive Learning graph for dense case}
%
%We provide a new proof for the recent result of Le Gall~\cite{gal14} and, along the way, we improve slightly the result by getting rid of the logarithmic factors. We show the following result.
%\begin{theorem}
%There exists a quantum algorithm that, given as input the oracle of an unweighted graph $G$ on $n$ vertices, outputs a triangle of $G$ with probability at least $2/3$ if a triangle exists, and uses $O(n^{5/4})$ queries to the oracle.
%\end{theorem}
%To prove this, we find an extended learning graph with complexity $O(n^{5/4})$ (see Theorem~\ref{thm:dense} below) and then apply Lemma~\ref{lem:ELG-querycomplexity}.

%\subsubsection{Le Gall's strategy}
We start by reviewing the main ideas of Le Gall's algorithm in order to find a triangle in an input graph $G$ with $n$ vertices. 
More precisely, we decompose the problem into similar subproblems, and we build up our adaptive learning graph on top of it.
Doing so, we get rid of most of the technical difficulties that arise in the resolution of the underlying problems using quantum walk based algorithms.

Let $V$ be the vertex set of $G$.
For a vertex $u$, let $N_{u}$ be the neighborhood of $u$, and for two vertices $u,v$, let $N_{u,v}=N_u\cap N_v$.
Figure~\ref{LeGall} illustrates the following strategy for finding a potential triangle in some given graph $G$.

First, fix an $x$-subset $X$ of vertices. Then, either $G$ has a triangle with one vertex in $X$ or each (potential) triangle vertex is outside $X$.
The first case is quite easy to deal with, so we ignore it for now and 
we only focus on the second case. Thus there is no need to query any possible edge
between two vertices $u,v$ connected to the same vertex in $X$. Indeed, if such an edge exists, the first case will detect a triangle.
Therefore one only needs to look for a triangle edge in
% potential edges in the set
%Indeed, otherwise there would be a triangle between $u,v$ and $X$. Therefore define 
$\Delta(X)=\{(u,v)\in V^2 : N_{u,v}\cap X=\emptyset\}$.
\begin{figure}[h]
\centerline{\includegraphics[height=3cm]{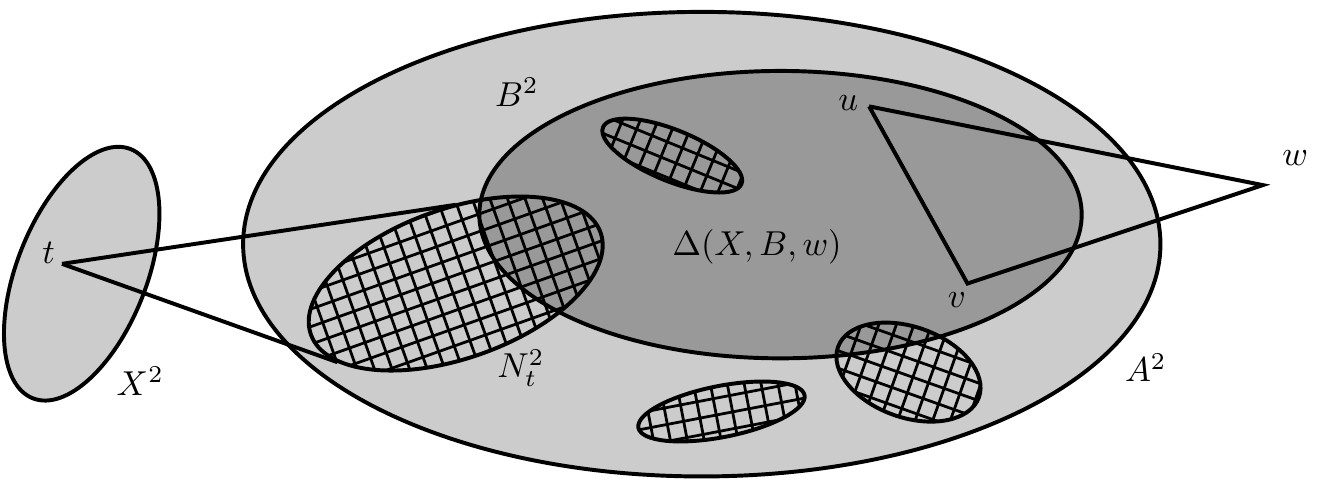}} %{FIGURES/Compact/Delta-XBw-compact-colors-label.pdf}} %{FIGURES/DeltaXBw-2-colored.png}}%{"Learning Graph - 1"}}
\caption{\label{LeGall}Sets involved in Le Gall's algorithm.}
\end{figure}

Second, search for an $a$-subset $A$ with two triangle vertices in it.
For this, construct the set $\Delta(X,A)=A^2\cap \Delta(X)$ of potential triangle edges in $A^2$.
The set $\Delta(X,A)$ can be easily set once all edges between $X$ and $A$ are known.

Third, in order to decide if $\Delta(X,A)$ has a triangle edge, search for a vertex $w$ making a triangle with an edge of $\Delta(X,A)$.
%If $w\in X$ we already know that we can reject. 

Otherwise, search for a $b$-subset $B$ of $A$ such that $w$ makes a triangle with two vertices of $B$.
For this last step, we construct the set $\Delta(X,B,w)=(N_{w})^2\cap \Delta(X,B)$ of pairs of vertices connected to $w$.
If any of such pairs is an actual edge, then we have found a triangle. 
%We search for such a pair using a Grover search.

We will use learning graphs of type $\GOR$ for the first step, for finding an appropriate vertex $w$, and for deciding weither $\Delta(X,B,w)$ has an edge;
and learning graphs of type $\GJohnson$ for finding subsets $A$ and $B$.

More formally now, let $\Triangle$ be the Boolean function such that $\Triangle(G)=1$ iff graph input $G$ has a triangle.
We do the following decomposition.
%\textbf{Decomposition of the function $f$: a sufficient condition for negative instances.}
First, observe that $\Triangle=\bigvee_{X\,:\, |X| = x} ( h_{X} \vee f_{X})$
with
%\begin{itemize}
	%\item 
	$h_{X}(G) = 1$ (resp. $f_{X}(G) = 1$) iff $G$ has a triangle with a vertex in $X$ (resp. with no vertex in $X$).
	%\item  $f_{X}(G) = 1$ iff $G$ has a triangle with no vertex in $X$.
%\end{itemize}
Then, we pursue the decomposition for $f_X(G)$ as $f_X(G)=\bigvee_{A\,:\,|A| = a} f_{X,A}(G)$ and 
$f_{X,A} (G)=\bigvee_{w \in V} f_{X,A,w}(G)$, for $A \subseteq V$ and $ w \in V$, where
\begin{itemize}
	\item $f_{X,A}(G) = 1$ iff $G$ has a triangle between two vertices in $A\setminus X$ and a third one outside $X$;
	\item $f_{X,A,w}(G) = 1$ iff $w\not\in X$ and $G$ has a triangle between $w$ and two vertices in $A\setminus X$.
\end{itemize}
Last, we can write $f_{X,A,w} (G)=\bigvee_{B\subset A,\; |B|=b} f_{X,B,w}(G)$. 

%$f_X(G)=\bigvee_{A \subseteq V, |A| = a} f_{X,A}(G)$ and 
%$f_{X,A} (G)=\bigvee_{w \in V} f_{X,A,w}(G)$.

%\subsubsection{A simpler algorithm using adaptative learning graphs}

With our notations introduced in Section~\ref{subsec:decomp-lemmas}, 
our adaptative learning graph $\G$ for Triangle Finding
%$\G$ consists in a $6$-stage learning graph composed by a Grover search on all $X$ followed by two branches. One for searching a vertex of triangle in $X$ and the other one for searching a triangle with all its vertices outside $X$.
%With our those notations, 
%$\G$ 
can be represented as in Figure~\ref{LeGallLG}. % where we use $\Load=\Loadall$.
%The resulting learning is therefore only adaptive.
\begin{figure}[h!]
\centerline{
\begin{tikzpicture}
\node (A) at (0,0.5) {$\emptyset$};
\node (B) at (2,0.5) {};
\draw[->,>=latex] (A) -- (B) node[above,midway,sloped] (1){$X$};
\node (C) at (3,1) {};
\node (D) at (3,0) {};
\draw[->,>=latex] (B) -- (C);
\draw[->,>=latex] (B) -- (D);
\node (E) at (4,0) {};
\node (F) at (5,0) {};
\draw[->,>=latex] (D) -- (E) node[above,midway,sloped] (2){$t$};
\draw[->,>=latex] (E) -- (F) node[above,midway,sloped] (3){$uv$};
\node (G) at (5,1){};
\node (H) at (5.5,1){};
\node (I) at (6,1){};
\draw[-,>=latex] (C) -- (I) node[above,midway,sloped] (4){\textit{$A$}};
\draw[->,>=latex] (C) -- (G);
\draw[->,>=latex] (G) -- (H);
\draw[->,>=latex] (H) -- (I);
\node (J) at (7,1) {};
\draw[->,>=latex] (I) -- (J) node[above,midway,sloped] (1){$w$};
\node (K) at (9,1){};
\node (L) at (9.5,1){};
\node (M) at (10,1){};
\draw[-,>=latex] (J) -- (M) node[above,midway,sloped] (4){\textit{$B$}};
\draw[->,>=latex] (J) -- (K);
\draw[->,>=latex] (K) -- (L);
\draw[->,>=latex] (L) -- (M);
\node (N) at (12,1) {};
\draw[->,>=latex] (M) -- (N) node[above,midway,sloped] (1){$\Delta(X,B,w)$};
\iffalse
\node[draw,circle,minimum height=0.4cm] (O) at (1,2.1) {$\scriptstyle 1$};
\node[draw,circle,minimum height=0.4cm] (P) at (2.5,2.1) {$\scriptstyle 2$};
\node[draw,circle,minimum height=0.4cm] (R) at (4.5,2.1) {$\scriptstyle 3$};
\node[draw,circle,minimum height=0.4cm] (T) at (6.5,2.1) {$\scriptstyle 4$};
\node[draw,circle,minimum height=0.4cm] (V) at (8.5,2.1) {$\scriptstyle 5$};
\node[draw,circle,minimum height=0.4cm] (X) at (11,2.1) {$\scriptstyle 6$};
\fi
\end{tikzpicture}
}\caption{\label{LeGallLG}Learning graph for Triangle Finding with complexity $O(n^{5/4})$.}
\end{figure}
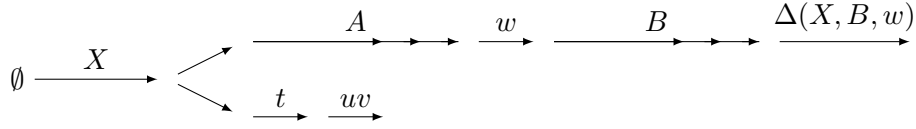

Using adaptive learning graphs instead of the framework of quantum walk based algorithms from~\cite{mnrs11}
simplifies the implementation of the above strategy because
one can consider all the possible subsets $X$ instead of choosing just a random one.
Then one only needs to estimate the average complexity over all possible $X$. 
Such an average analysis was not considered in the framework of~\cite{mnrs11}.
In addition, we do not need to estimate the size of $\Delta(X,A,w)$ at any moment of our algorithm.
As a consequence, our framework greatly simplifies the combinatorial analysis of our algorithm as compared to the one of Le~Gall,
and lets us shave off some logarithmic factors. See Appendix~\ref{app:thm:dense} for the proof of the following theorem.

\begin{theorem}\label{thm:dense}
The adaptive learning graph of Figure~\ref{LeGallLG}
 with $|X|=x$, $|A|=a$,  $|B|=b$, and using $\Load=\Loadall$,
 has complexity
$$
O\left(\sqrt{	xn^2
	+ (ax)^2 
	+\left(\frac{n}{a}\right)^2\left(a\cdot x^2 
		+ {n}\left(b^2
			+\left(\frac{a}{b}\right)^2\left( {b}
			+ \frac{b^2}{x} \right)
		\right)
	\right)}\right).
$$
In particular, taking $a=n^{3/4}$ and $b=x=\sqrt{n}$ leads to $Q(\Triangle)=O(n^{5/4})$.
\end{theorem}

%\subsubsection{Complexity analysis}

% SSSSSSSSSSSSSSSSSSSSSSSSSSSSSSSSSSSSSSSSSSSSSSSSSSSSSSSSSSSSSSSSSSSSSSSSSSSSSSSSSSSSSSSSSSSSSSSSSSSSSSSSSSSSSSSSSSSSSSSSSSSSSSSSSSSSSSSSSSSS %
\subsection{Sparse graphs}
In the sparse case we now show to use extended learning graphs in order to get a better complexity than the one of Theorem~\ref{thm:dense}.
Proofs of this section are deferred to Appendix~\ref{app:thm:sparse}.
%\subsection{A better algorithm using extended learning graphs}

First, the same learning graph of Theorem~\ref{thm:dense}
has a much smaller complexity for sparse graphs when $\Loadsparse{}$ is used instead of $\Loadall{}$.
\begin{theorem}\label{thm:sparse}
The learning graph of Figure~\ref{LeGallLG},
using $\Load=\Loadsparse$, has complexity over graphs with $m$ edges 
%and average degree at most $d=2m/n$
$$
{O}\left(\sqrt{	
\left(xm+(ax)^2\cdot\frac{m}{n^2}+\left(\frac{n}{a}\right)^2\left(a\cdot x^2 \cdot\frac{m}{n^2} + {n}\left(b^2\cdot\frac{m}{n^2}+\left(\frac{a}{b}\right)^2\left( {b}+ \frac{b^2}{x} \right) \right) \right)\right)\log n
}\right).
$$
In particular,  taking $a=n^{3/4}$ and $b=x=\sqrt{n}/(m/n^2)^{1/3}$ leads to a complexity of ${O}(n^{11/12}m^{1/6}\sqrt{\log n})$ when $m\geq n^{5/4}$.
\end{theorem}

We now end with an even simpler learning graph whose complexity depends on its average of squared degrees.
It is very simple. 
See Figure~\ref{SparseNewDec} for the illustration.
It consists in searching for a triangle vertex $w$.
In order to check if $w$ is such a vertex, we search for a $b$-subset $B$ with an edge connected to $w$.
For this purpose, we first connect $w$ to $B$, and then check if there is an edge in $(N_w\cap B)^2$.

Formally, we do the decomposition
$\Triangle=\bigvee_{w \in V} f_w$,
with $f_w(G) = 1$ iff $w$ is a triangle vertex in $G$. Then, we pursue the decomposition with $f_w(G) = \bigvee_{B \subseteq V \,:\,|B| = b} f_{w,B}(G)$
where $f_{w,B}(G) = 1$ iff $G$ has a triangle formed by $w$ and two vertices of $B$.
Using our notations, the resulting learning graph is represented by the diagram in Figure~\ref{SparseNewLG}.

\begin{figure}[h]
\centerline{\includegraphics[height=3cm]{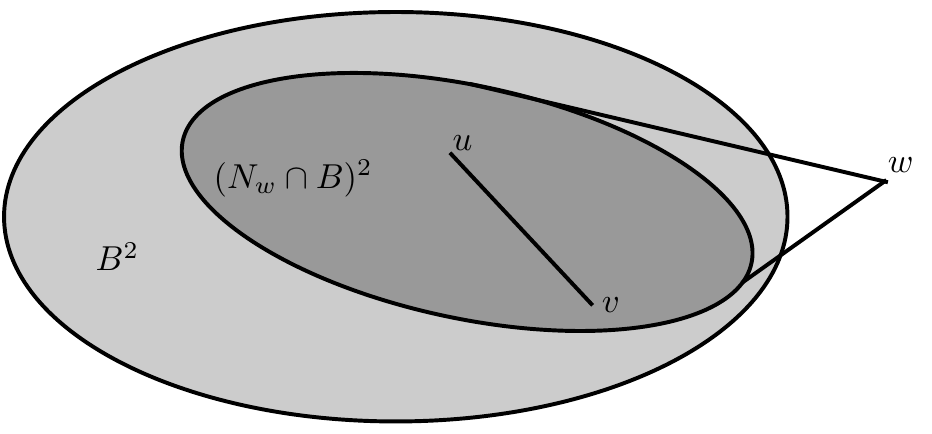}} %{FIGURES/neighbors-w-2.eps}}%{"Learning Graph - 2"}}
\caption{\label{SparseNewDec}Sets involved in the sparse decomposition.}
\end{figure}
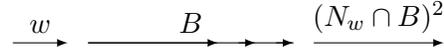
\begin{figure}[h]
\centerline{
\begin{tikzpicture}
\node (J) at (7,1) {};
\draw[->,>=latex] (I) -- (J) node[above,midway,sloped] (1){$w$};
\node (K) at (9,1){};
\node (L) at (9.5,1){};
\node (M) at (10,1){};
\draw[-,>=latex] (J) -- (M) node[above,midway,sloped] (4){\textit{$B$}};
\draw[->,>=latex] (J) -- (K);
\draw[->,>=latex] (K) -- (L);
\draw[->,>=latex] (L) -- (M);
\node (N) at (12,1) {};
\draw[->,>=latex] (M) -- (N) node[above,midway,sloped] (1){$(N_w\cap B)^2$};
\iffalse
\node[draw,circle,minimum height=0.4cm] (T) at (6.5,2.1) {$\scriptstyle 1$};
\node[draw,circle,minimum height=0.4cm] (V) at (8.5,2.1) {$\scriptstyle 2$};
\node[draw,circle,minimum height=0.4cm] (X) at (11,2.1) {$\scriptstyle 3$};
\fi
\end{tikzpicture}
}\caption{\label{SparseNewLG}Learning graph for Triangle Finding with complexity $\widetilde{O}((n^{5/6}m^{1/6}+d_2\sqrt{n})\log n)$.}
\end{figure}
%
%\begin{figure}[h]
%\centerline{\includegraphics[width=7cm]{FIGURES/neighbors-w-2.eps}}
%%\caption{\label{SparseNewLG}}
%\end{figure}

We prove in Appendix~\ref{app:thm:sparse} the following theorem, where $d_2 = \sqrt{\Exp_v \big[\,|N_v|^2\,\big]}$ denotes the variance of the degrees.
\begin{theorem}\label{thm:sparsenew}
Let $b\geq n^2/m$.
The learning graph of Figure~\ref{SparseNewLG}, using 
$\Loadsparse$ for the first stage of $\GJohnson$ and
$\Loadall$ otherwise, 
has complexity over graphs with  $m$ edges 
%and average degree at most $d=2m/n$
$$
{O}\left(\sqrt{	
n\left(b^2 \frac{m}{n^2}\log n
+\frac{n^2}{b^2}\left(
b+\frac{b^2(d_2)^2 }{n^2}
\right)\right)
}\right).
$$
Taking $b=n^{4/3}/(m\log n)^{1/3}$ %(which is $\geq n^2/m$)
leads to a complexity of ${O}(n^{5/6}(m\log n)^{1/6}+d_2\sqrt{n})$.
\end{theorem}

\bibliography{lgtriangle}

% AAAAAAAAAAAAAAAAAAAAAAAAAAAAAAAAAAAAAAAAAAAAAAAAAAAAAAAAAAAAAAAAAAAAAA %
\appendix

% SSSSSSSSSSSSSSSSSSSSSSSSSSSSSSSSSSSSSSSSSSSSSSSSSSSSSSSSSSSSSSSSSSSSSSSSSSSSSSSSSSSSSSSSSSSSSSSSSSSSSSSSSSSSSSSSSSSSSSSSSSSSSSSSSSSSSSSSSSSS %
\section{Learning graph analysis for dense graphs}
\label{app:thm:dense}
% SSSSSSSSSSSSSSSSSSSSSSSSSSSSSSSSSSSSSSSSSSSSSSSSSSSSSSSSSSSSSSSSSSSSSSSSSSSSSSSSSSSSSSSSSSSSSSSSSSSSSSSSSSSSSSSSSSSSSSSSSSSSSSSSSSSSSSSSSSSS %
\begin{proof}[Proof of Theorem~\ref{thm:dense}]
From now on, fix some input graph $G = (V,E)$ (with or without a triangle).
%When $G$ has at least one triangle, we name $\alpha,\beta,\gamma$ the vertices of one of its triangles.
We compute the complexity $C(\G, G)$ of $\G$ on $G$ using Lemmas~\ref{groverExtended} and~\ref{johnsonExtended}.
From the decomposition of $\Triangle$ one can already check that the resulting learning graph computes the function $\Triangle$.

%Let $f$ be the boolean function such that $f(G)=1$ iff $G$ has a triangle.
%From now on, we replace all $O()$ notations by inequalities $\leq$ where constant multiplicative factors are omitted. 
Also all complexities for positive instances will be at most $1$. Therefore, we only compute the complexity of negative instances, and drop multiplicative factors corresponding to the complexity of a learning graph on positive instances. 

We decompose the analysis in stages as in Figure~\ref{LeGallLG2},
and we compute their respective negative complexities on some given graph $G$.
\begin{figure}[h!]
\centerline{
\begin{tikzpicture}
\node (A) at (0,0.5) {$\emptyset$};
\node (B) at (2,0.5) {};
\draw[->,>=latex] (A) -- (B) node[above,midway,sloped] (1){$X$};
\node (C) at (3,1) {};
\node (D) at (3,0) {};
\draw[->,>=latex] (B) -- (C);
\draw[->,>=latex] (B) -- (D);
\node (E) at (4,0) {};
\node (F) at (5,0) {};
\draw[->,>=latex] (D) -- (E) node[above,midway,sloped] (2){$t$};
\draw[->,>=latex] (E) -- (F) node[above,midway,sloped] (3){$uv$};
\node (G) at (5,1){};
\node (H) at (5.5,1){};
\node (I) at (6,1){};
\draw[-,>=latex] (C) -- (I) node[above,midway,sloped] (4){\textit{$A$}};
\draw[->,>=latex] (C) -- (G);
\draw[->,>=latex] (G) -- (H);
\draw[->,>=latex] (H) -- (I);
\node (J) at (7,1) {};
\draw[->,>=latex] (I) -- (J) node[above,midway,sloped] (1){$w$};
\node (K) at (9,1){};
\node (L) at (9.5,1){};
\node (M) at (10,1){};
\draw[-,>=latex] (J) -- (M) node[above,midway,sloped] (4){\textit{$B$}};
\draw[->,>=latex] (J) -- (K);
\draw[->,>=latex] (K) -- (L);
\draw[->,>=latex] (L) -- (M);
\node (N) at (12,1) {};
\draw[->,>=latex] (M) -- (N) node[above,midway,sloped] (1){$\Delta(X,B,w)$};
\iftrue
\node[draw,circle,minimum height=0.4cm] (O) at (1,2.1) {$\scriptstyle 1$};
\node[draw,circle,minimum height=0.4cm] (P) at (2.5,2.1) {$\scriptstyle 2$};
\node[draw,circle,minimum height=0.4cm] (R) at (4.5,2.1) {$\scriptstyle 3$};
\node[draw,circle,minimum height=0.4cm] (T) at (6.5,2.1) {$\scriptstyle 4$};
\node[draw,circle,minimum height=0.4cm] (V) at (8.5,2.1) {$\scriptstyle 5$};
\node[draw,circle,minimum height=0.4cm] (X) at (11,2.1) {$\scriptstyle 6$};
\fi
\end{tikzpicture}
}\caption{\label{LeGallLG2}Adaptive learning graph for Triangle Finding with its corresponding stages.}
\end{figure}
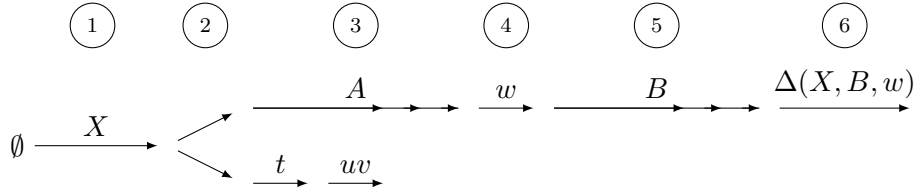
% PPPPPPPPPPPPPPPPPPPPPPPPPPPPPPPPPPP %
%\medbreak
%\noindent
%\textbf{Decomposition of the function $f$: a sufficient condition for negative instances.}
%
%We can decompose $f$ as:
%$$f=\bigvee_{X \subseteq V, \, |X| = x} g_{X}$$
%where $g_{X} = h_{X} \vee f_{X}$ and:
%\begin{itemize}
%	\item $h_{X} = 1$ iff $G$ has a triangle with a vertex in $X$
%	\item $f_{X} = 1$ iff $G$ has a triangle with no vertex in $X$.
%\end{itemize}
%We will also use the following notations, for $A \subseteq V, w \in V$:
%\begin{itemize}
%	\item $f_{X,A} = 1$ iff $G$ has a triangle between two vertices in $A$ and a third one outside $X$
%	\item $f_{X,A,w} = 1$ iff $G$ has a triangle between $w$ and two vertices in $A$ 
%\end{itemize}
%We have, for all $X$ and $a$,
%$$
%	\bigvee_{A \subseteq V, |A| = a} f_{X,A} = 0 \quad\Rightarrow\quad f_X = 0
%$$
%and for all $X$ and all $A$,
%$$
%	\bigvee_{w \in V} f_{X,A,w} = 0 \quad\Rightarrow\quad f_{X,A} = 0.
%$$
%In other words, to decide if $G$ is a negative instance ($f = 0$), it is enough to check if
%$$
%	\bigvee_{X \subseteq V, \, |X| = x} h_{X} \vee \left( \bigvee_{A \subseteq V, |A| = a} \bigvee_{w \in V} f_{X,A,w} \right) = 0.
%$$
%This condition will be easier to check in the analysis.
%
%% PPPPPPPPPPPPPPPPPPPPPPPPPPPPPPPPPPP %
%\medbreak

\noindent
\textbf{Stage 1.}
$\G$ consists in the combination of learning graphs $\G_X$, where $X$ is a $x$-subset of $[n]$, as in Lemma~\ref{groverExtended}.
The particularity of $\G_X$ is that they all compute $\Triangle$.

% PPPPPPPPPPPPPPPPPPPPPPPPPPPPPPPPPPP %
\noindent
\textbf{Stage 2.}
Each $\G_X$ is again a combination of two learning graphs $\F_{X}$ and $\HH_{X}$ as in Lemma~\ref{groverExtended}.
The learning graph $\F_X$, described in the remaining stages, computes $f_X$, whereas $\HH_X$ computes $h_X$.
Observe that $\HH_X$ consists in a very simple non-adaptative learning graph with negative complexity $xn^2$.
Therefore we can already deduce that
$$
	C^0(\G, G)
	\leq 
	\Exp_{X}(C^0(\F_X,G)) +xn^2.
$$

In the sequel we focus on the analysis of $\F_X$.

%{\color{blue}
%At Stage~3, $\F_{X}$ is decomposed using Lemma~\ref{johnsonExtended} with parameters $n'=n$, $k'=a$, $r'=2$
%and $\mathrm{\bf{S}}=ax$, $\mathrm{\bf{U}}=x$. Therefore,
%when $G$ has no triangle between a vertex outside $X$,
%
%{\color{red} UTILISER LE COROLLAIRE AVEC complexit\' de LEARN}
%$$C^0(\E_{X}, G) \leq
%(ax)^2 + \left(\frac{n}{a}\right)^2\left({a}\cdot x^2 + \Exp_{A}(C^0(\E_{X,A}, G))      \right),
%$$
%where $A$ is a $a$-subset of $[n]$,
%$f_{X,A}(G)=1$ iff $G$ has a triangle between two vertices in $A$ and a third one outside $X$,
%and $C^0(\E_{X,A}, G)$ is the complexity of a learning graph for $f_{X,A}$ in $G$.
%We could have enforce that $A\cap X=\emptyset$, but we prefer to make $A$ and $X$ independent for simplifying some further analysis.
%}

% PPPPPPPPPPPPPPPPPPPPPPPPPPPPPPPPPPP %
\noindent
\textbf{Stage 3.}
$\F_{X}$ is decomposed using Lemma~\ref{johnsonExtended} with $\Learn_{S} = \Loadall_{S}$ and parameters $n=|V|$, $k=a$, $r=2$. Therefore
$$
	C^0(\E_{X}, G)
	=
	O\left((\mathrm{\bf{S}}_{G,X})^2+\left(\frac{n}{a}\right)^{2}\left(a \cdot (\mathrm{\bf{U}}_{G,X})^2+ (\mathrm{\bf{C}}_{G,X})^2\right)\right)
$$
where we take
\begin{align*}
		&(\mathrm{\bf{S}}_{G,X})^2 = {\Exp_{A' \subset V \,:\, |A'| = a-2}  \left( |I_X(A')|^2 \right)},
		\\
		&(\mathrm{\bf{U}}_{G,X})^2 = \max_{a-2 \leq i <  a}  \left({ \Exp_{A' \subset A'' \subseteq V \,:\, |A'| = |A''|-1=i}  \left(|I_X(A'')\backslash I_X(A')|^2 \right)}\right),
		\\
		&(\mathrm{\bf{C}}_{G,X})^2 = {\Exp_{A\subseteq V \,:\, |A| = a} \left(C^0(\F_{X,A}, G)\right)},
\end{align*}
with $I_X(A') = X \times A$, $G_{I_X(A')}$ is the set of edges between $X$ and $A'$ in the graph $G$, and $\F_{X,A}$ is the learning graph for $f_{X,A}$
that we describe in the remaining stages.
%{\color{red} *** We could have enforced that $A\cap X=\emptyset$, but we prefer to make $A$ and $X$ independent for simplifying some further analysis. ***}

% PPPPPPPPPPPPPPPPPPPPPPPPPPPPPPPPPPP %
\noindent
\textbf{Stage 4.}
We use Lemma~\ref{groverExtended} and $w$ stands for the third triangle vertex.
Therefore, 
$$
	C^0(\F_{X,A}, G) \leq 
	{n} \times \Exp_{w}(C^0(\F_{X,A,w}, G)),
$$
where $w\in[n]$ and $\F_{X,A,w}$ is the learning graph for $f_{X,A,w}$ described below.
%{\color{red} *** Again we could have enforce that $w\not\in X$.}

% PPPPPPPPPPPPPPPPPPPPPPPPPPPPPPPPPPP %
\noindent
\textbf{Stage 5.}
Next, we use Lemma~\ref{johnsonExtended} with $\Learn_{S} = \Loadall_{S}$ and  parameters  $n'=a$, $k'=b$, $r'=2$. Therefore
%when $G$ has no triangle between $w$ and two vertices in $A$,
%$$C^0(\E_{X,A,w}, G) \leq 
%b^2 + \left(\frac{a}{b}\right)^2 \left({b} + \Exp_{B}(C^0(\E_{X,A,w,B}, G))      \right),
%$$
$$
	C^0(\F_{X,A,w}, G)
	=
	O\left((\mathrm{\bf{S}}_{G,X,A,w})^2+\left(\frac{a}{b}\right)^{2}\left(b \cdot (\mathrm{\bf{U}}_{G,X,A,w})^2+ (\mathrm{\bf{C}}_{G,X,A,w})^2\right)\right)
$$
where we take
\begin{align*}
		&(\mathrm{\bf{S}}_{G,X,A,w})^2 = {\Exp_{B' \subset A\, :\, |B'| = b-2}  \left( |I_{w}(B')|^2 \right)},
		\\
		&(\mathrm{\bf{U}}_{G,X,A,w})^2 = \max_{b-2 \leq i < b}  \left({ \Exp_{B' \subset B'' \subseteq V\,:\, |B'| = |B''|-1=i}  \left(|I_{w}(B'')\backslash I_{w}(B')|^2 \right)}\right),
		\\
		&(\mathrm{\bf{C}}_{G,X,A,w})^2 = {\Exp_{B \subseteq A \,:\, |B| = b} \left(C^0(\E_{X,A,w,B}, G)\right)},
\end{align*}
%\begin{align*}
%		&\mathrm{\bf{S}}_{G,X,A,w} = \sqrt{\Exp_{B \subseteq A \,:\, |B| = b-2}  \left( |I_{w}(B)| \cdot \log(|I_{w}(B)| +1) \cdot (|G_{I_{w}(B)}|+1) \right)},
%		\\
%		&\mathrm{\bf{U}}_{G,X,A,w} = \max_{B \subseteq A \,:\,1 \leq |B| < b-2}  \sqrt{ \Exp_{B' \subset B \,:\, |B'| = |B|-1}  \left(|I_{w}(B)\backslash I_{w}(B')| \cdot \log(|I_{w}(B)\backslash I_{X,A}(B')| +1) \cdot (|G_{I_{w}(B)\backslash I_{w}(B')}|+1) \right)},
%		\\
%		&\mathrm{\bf{C}}_{G,X,A,w} = \sqrt{\Exp_{B \subseteq A \,:\, |B| = b} \left(C^0(\E_{X,A,w,B}, G)\right)},
%\end{align*}
with $I_{w}(B') = \{w\} \times B'$, 
$G_{I_{w}(B')}$ is the set of edges linking $w$ to $B'$ in the graph $G$, 
%$f_{X,A,w,B} = f_{X,B,w}$,
and $\F_{X,A,w,B}$ is the learning graph for $f_{X,A,w,B}$ described in the last stage.

% PPPPPPPPPPPPPPPPPPPPPPPPPPPPPPPPPPP %
\noindent
\textbf{Stage 6.}
The last stage consists in the learning graph obtained by Lemma~\ref{groverExtended}, with negative complexity of order $|\Delta(X,B,w)|$,
for searching a potential edge in $\Delta(X,B,w)$.\smallskip
%$f_{X,A,w,B} = f_{X,B,w}$,
%for each graph $G$ with no (triangle) edge in $\Delta(X,B,w)$.

In order to conclude, we observe that
%we have the following bound %(without any assumption on the sparsity of $G$): 
for any $w \in V$ and any set of vertices $V_1 \subseteq V$, we have $|I_X(V_1)| = x |V_1|$ and  $|I_w(V_1)| = |V_1|$. 
Applying this for $V_1=A$ and $V_1=B$ we obtain
\begin{align*}
		\mathrm{\bf{S}}_{G,X} \leq ax  ,
		\qquad
		\mathrm{\bf{U}}_{G,X} \leq x ,
		\qquad
		\mathrm{\bf{S}}_{G,X,A,w} \leq b ,
		\qquad
		\mathrm{\bf{U}}_{G,X,A,w} \leq 1.
\end{align*}

We therefore get that $C^0(\G, G)$ has order
$$
	xn^2
	+ (ax)^2 
	+\left(\frac{n}{a}\right)^2\left(a\cdot x^2 
		+ {n}\left(b^2
			+\left(\frac{a}{b}\right)^2\left( {b}
			+ \Exp_{X,w,B} \big[\,|\Delta(X,B,w)|\,\big] \right)
		\right)
	\right).
$$
We now conclude using  Lemma~\ref{delta} with $V_1=V$, and $C(\G) = \sqrt{C^0(\G)}$ since $C^1(\G) \leq 1$.
\end{proof}

\begin{lemma}\label{delta}
%	{\color{red} *** COMBINATORIAL LEMMA --- A RETIRER ?? ***}	
Let $x,b$ be positive integers.
Let $G$ be a graph on a vertex set $V$ of size $n$ and let $B\subseteq V$ be a $b$-subset.
%First observe that Lemma~\ref{delta} remains valid if $X$ is restricted to some subset $V_1$ of vertices. 
Then
$$\Exp_{X,w} \big[\,|\Delta(X,B,w)|\,\big] \le\frac{b^{2}}{x},$$
where the expectation is taken over $x$-subsets $X\subseteq V$
and  vertices $w\in V$.
\end{lemma}
\begin{proof}
Let $\Delta(X)$ be the set of pairs of vertices which are not both neighbors of any vertex in $X$.
				Let $B\subseteq V$ of size $b$, the expectation on $X$ and $w$ is:
\begin{equation}\label{expd}
\Exp_{X,w}\big[\,|\Delta(X,B,w)|\,\big]=\sum_{(u,v)\in B^{2}}\Pr_{X,w}((u,v)\in \Delta(X,B,w)).
\end{equation}
			
In order to bound the probability events of the right hand side, fix $(u,v)\in B^{2}$ and let
$N_{u,v}$ be the intersection of the neighborhoods of $u$ and $v$ in $V$. 
Then
$$\Pr_{X,w}((u,v)\in \Delta(X,B,w))=\Pr_{X,w}(w\in N_{u,v} \text{ and } (u,v)\in \Delta(X)).$$
The two events of the right hand side are independent, therefore with $t=|N_{u,v}|$ 
and $n=|V|$
we get
$$\Pr_{X,w}((u,v)\in \Delta(X,B,w))=\frac{t}{n}\left(1-\frac{t}{n}\right)^{x}.
$$
Renaming $\alpha=\frac{tx}{n}$ leads to
$$\Pr_{X,w}((u,v)\in \Delta(X,B,w))=\frac{\alpha}{x}\left(1-\frac{\alpha}{x}\right)^{x}\le \frac{\alpha e^{-\alpha}}{x}\le \frac{1}{x}.
$$

Finally, combining the above bound with equation~\eqref{expd} gives the result.
%$$\Exp_{X,w}\big[\,|\Delta(X,B,w)|\,\big]\le \frac{b^{2}}{x}.$$
\end{proof}

\section{Learning graph analyses for sparse graphs}\label{app:thm:sparse}
\subsection{Proof of Theorem~\ref{thm:sparse}}

%\begin{proof}[Proof of Theorem~\ref{thm:sparse}]

We reuse the notations introduced in the proof of Theorem~\ref{thm:dense}.
In addition we let $d=2m/n$ the average degree of the input graph.

Here $\HH_{X}$ has complexity $O(x dn\log n + x (d_2)^2) = O(x d n\log n)$ on any negative instance,
where $d_2=\sqrt{\Exp_v (|N_v|^2)}$. Indeed, fix any negative instance.
For each $v \in X$, we learn $N_v$ by loading $\{v\}\times V$ with negative complexity $|N_v| (n+1) \log (n +1)$.
Then we load $N_v\times N_v$ simply using $\Loadall$ with negative complexity $|N_v|^2$.
So summing those complexities for every $v\in X$ and taking the expectation on $x$-subsets $X \subseteq V$, 
the average negative complexity becomes
\begin{eqnarray*}
	&&\Exp_{X\subseteq V, |X|=x} \left[ \sum_{v \in X} \left(|N_v| (n+1) \log (n +1) + |N_v|^2\right)\right] \\
%	&=
%	xn +  \Exp_{X\subseteq V, |X|=x} \left[ \sum_{v \in X} |N_v|^2\right]
%	\\
%	&=
%	xn +  \sum_{X\subseteq V, |X|=x} {\frac{1}{{n \choose x}}} \sum_{v \in X} |N_v|^2
%	\\
%	&=
%	xn +  \sum_{v \in V} {\frac{1}{{n \choose x}}} \sum_{X \ni v, |X|=x} |N_v|^2
%	\\
%	&=
%	xn +  {\frac{{(n-1) \choose (x-1)}}{{n \choose x}}}  n \Exp_{v \in V} \left[ |N_v|^2 \right]
&=&	x \times \Exp_{v \in V} \left( |N_v| (n+1) \log (n +1) + |N_v|^2 \right)
	\\
	&=&
	xd(n+1)\log (n+1) +  x  (d_2)^2.
\end{eqnarray*}
%Moreover $(d_2)^2 = \Exp_v\left[|N_v|^2\right] \leq \Exp_v\left[|N_v|\right] n= d n$.

For the first step of the Johnson walk in $\F_X$, we now get using $\Loadsparse$
\begin{align*}
		\Exp_X \left( (\mathrm{\bf{S}}_{G,X})^2 \right)
		&= \Exp_X \Exp_{A' \subseteq V\,:\, |A'| = a-2}  \left( |I_X(A')| \cdot \log(|I_X(A')|+1) \cdot (|G_{I_X(A')}|+1) \right) \\
		&\leq ax \log(ax+1) \Exp_X \Exp_{A' \subseteq V \,:\, |A'| = a-2} (|E(X,A')|+1)  \\
		&= O\left(\left(\frac{ax}{n}\right)^2 m \log(ax+1) \right),
\end{align*}
where for the second step we used  that $|I_X(A')| = |X \times A'| \leq ax$, and for the last one Lemma~\ref{ninterExpExp} below with $X$ and $Y= A$. Similarly, using the fact that $|I_X(A'')\backslash I_X(A')| = |X \times (A'' \backslash A')| = x$, we obtain 
\begin{align*}
		\Exp_X \left(  (\mathrm{\bf{U}}_{G,X})^2 \right)
%		\\
%		&= \max_{1 \leq i < a-2}  \Exp_X \Exp_{A  \subseteq V \,:\, |A| =i}\Exp_{A' \subset A \,:\, |A'| = |A|-1}  \left(|I_X(A)\backslash I_X(A')| \cdot \log(|I_X(A)\backslash I_X(A')|+1) \cdot (|G_{I_X(A)\backslash I_X(A')}|+1) \right)
		&= \Exp_X\left(x \cdot \log(x+1) \cdot  \max_{a-2 \leq i < a} \left(\Exp_{A'  \subset V \,:\, |A'| =i}\left(\Exp_{v \in V \backslash A'}   (|E(X,v)|+1) \right)\right)\right)\\
		&= O\left(\Exp_X\left(x \cdot \log(x+1) \cdot  \left(\Exp_{v \in V}   (|E(X,v)|+1) \right)\right)\right)	\\
%		&= O\left( x \cdot \log(x+1) \cdot  \max_{1 \leq i < a-2} \left(\Exp_{A  \subseteq V \,:\, |A| =i-1}\Exp_{\alpha \in V \backslash A}  \left( x \frac{|N_\alpha|}{n} \right) 		\right)\right) \\
%		&= O\left( \frac{x^2}{n} \cdot \log(x+1) \cdot  \Exp_{v \in V }  (|N_v|)		\right)		\\
		&= O\left( \frac{x^2 m}{n^2} \cdot \log(x+1)
		\right),
%		\\
%		&\leq O\left(2 \begin{pmatrix}n-1\\ i-1\end{pmatrix}^{-1} \frac{x}{n^2} m \log(x+1) \right) \tag{\color{red}Not clear: CHECK AGAIN}
\end{align*}
where the second equality holds by Lemma~\ref{ninterExpExp} below with $X$ and $Y =\{v\}$.
% and $N=N_\alpha$.

For the second step of the walk, since $|I_{w}(B')| = |\{w\} \times B'| \leq b$, we have
by Lemma~\ref{lem:ninter}, with $x = b-2$, $V_1= A$ and $N=N_w\cap A$:
\begin{align*}
		(\mathrm{\bf{S}}_{G,X,A,w})^2 
%		&= \Exp_{B \subseteq A \,:\, |B| = b-2}  \left( |I_{w}(B)| \cdot \log(|I_{w}(B)|+1) \cdot (|G_{I_{w}(B)}|+1) \right)
%		\\
		&\leq b \log(b+1) \Exp_{B' \subset A \,:\, |B'| = b-2}  \left( |E(B',w)|+1 \right)
%		\\
%		&= b \log(b+1) \left( \frac{(b-2)  |N_w \cap A |}{a} +1 \right)
		\\
		&= O\left(\frac{b^2  |N_w\cap A |}{a} \log(b+1)  \right).
\end{align*}
%where the second equality holds we used by Lemma~\ref{lem:ninter}, with $x = b-2$, $V_1= A$ and $N=N_w\cap A$.
 Moreover, again by Lemma~\ref{lem:ninter} below but this time with $x = a$, $V_1= V$ and $N=N_w$, we get
$$
	\Exp_{w \in V} \Exp_{A \subseteq V \,:\, |A| = a} \left( |N_w\cap A | \right) = \frac{a}{n} \Exp_{w \in V} |N_w \cap V| = \frac{a d}{n}. 
$$
So,
$$
	\Exp_{A \subseteq V \,:\, |A| = a} \left( (\mathrm{\bf{S}}_{G,X,A,w})^2  \right) = O\left(b^2   \cdot \frac{d}{n} \log(b+1)  \right). 
$$

Last, since $|I_{w}(B'')\backslash I_{w}(B')| = |\{w\} \times (B'' \backslash B')| = 1$, we directly obtain:
\begin{align*}
		\mathrm{\bf{U}}_{G,X,A,w}^2 
%		&= \max_{1 \leq i < b-2}  \Exp_{B \subseteq A \,:\, |B| = i} \Exp_{B' \subset B \,:\, |B'| = |B|-1}  \left(|I_{w}(B)\backslash I_{w}(B')| \cdot \log(|I_{w}(B)\backslash I_{X,A}(B')|+1) \cdot (|G_{I_{w}(B)\backslash I_{w}(B')}|+1) \right)
%		\\
%		&= \left( \max_{b-2 \leq i < b}  \left(\Exp_{B' \subseteq A \,:\, |B'| = i} 
%		\left(\Exp_{v \in A \backslash B'}  \left( |E(w, v)|+1 \right)\right)\right)\right)
		&= O( 1 ).
\end{align*}
Thus, the total negative complexity is of order
\begin{equation}
\label{eq:zerocomp-sparse}
	\underbrace{
	\left[
	xdn+(ax)^2\cdot\frac{d}{n}+\left(\frac{n}{a}\right)^2\left(a\cdot x^2 \cdot\frac{d}{n} + {n}\left(b^2\cdot\frac{d}{n}+\left(\frac{a}{b}\right)^2\left( {b}+ \frac{b^2}{x} \right) \right) \right)
	\right]
	}_{\displaystyle K_n}
	\times
	\log(n).
\end{equation}
%{\color{red}
%$$
%	n^2+xdn+(ax)^2\cdot\frac{d}{n} \log(ax+1)+\left(\frac{n}{a}\right)^2\left(a\cdot x^2 \cdot\frac{d}{n}  \log(x+1) + {n}\left(\frac{b^2}{a} d+\left(\frac{a}{b}\right)^2\left( {b}+ \frac{b^2}{x} \right) \right) \right),
%$$
%}
Denoting $t=\frac{d}{n}\leq 1$, we have:
$$
	K_n
	=
	n^2+t\left(xn^2+a^2x^2+\frac{n^2}{a^2}\left(ax^2+nb^2\right)\right)
+\frac{n^3}{b}\left( 1+\frac{b}{x}\right).
$$
%{\color{red}
%$$
%n^2+t\left(xn^2+a^2x^2 \log(ax+1) +\frac{n^2}{a^2}\left(ax^2 \log(x+1) +nb^2\right)\right)
%+\frac{n^3}{b}\left( 1+\frac{b}{x}\right).
%$$
%}
%Recall that  $b \leq a\leq n$.
%If $x\leq b$, we get:
%$$
%K_n \in O\left(
%t\left(xn^2+a^2x^2+\frac{n^3 b^2}{a^2} \right)
%+\frac{n^3}{x}
%\right).
%$$
%whereas if $a\geq x\geq b$, we have $xn^2\geq x^2n^2/a$, hence:
%$$
%K_n \in O\left(
%t\left(xn^2+a^2x^2+\frac{n^3 b^2}{a^2} \right)
%+\frac{n^3}{b}
%\right).
%$$
%Taking $\sqrt{n}\leq b=x\leq a$ and $a=n^{3/4}$ leads to
%$$
%t b^2n^{3/2}
%+\frac{n^3}{b}.
%$$
%Setting $b=\sqrt{n}/t^{1/3}\geq \sqrt{n}$ gives a $0$-complexity of $n^{5/2}t^{1/3}$, and
%a global complexity of $n^{5/4}t^{1/6}=n^{13/12}d^{1/6}=n^{11/12}m^{1/6}$.
%This complexity is valid until $\sqrt{n}/t^{1/3}\leq n^{3/4}$, that is when $t\geq 1/n^{3/4}$, \textit{i.e.} $d\geq n^{1/4}$.

If $x = b \leq a \leq n$, we have $xn^2\geq x^2n^2/a$ and $\frac{n^3}{b} \geq n^2$, hence :
$$
	K_n =
	O\left(
	t\left(xn^2+a^2x^2+\frac{n^3 b^2}{a^2} \right)
	+\frac{n^3}{b}
	\right). 
$$
Taking $a=n^{3/4}$ and $b = x = \sqrt{n}/t^{1/3}$, leads to:
$$
	K_n = O\left(
	t b^2n^{3/2}
	+\frac{n^3}{b}
	\right)
	=
	O\left(
	n^{5/2}t^{1/3}
	\right) . 
$$
Going back to~\eqref{eq:zerocomp-sparse}, this yields a negative complexity of order:
$$
	n^{5/2}t^{1/3}
	\times
	\log(n),
$$
and thus a total complexity of order:
$$
	n^{5/4}t^{1/6} \cdot \log(n)^{1/2} = n^{11/12}m^{1/6} \log(n)^{1/2}.
$$
This complexity is valid until $\sqrt{n}/t^{1/3}\leq n^{3/4}$, that is when $t\geq 1/n^{3/4}$, \textit{i.e.} $d\geq n^{1/4}$.
This concludes the proof of the theorem.

%\end{proof}

%\subsubsection{Useful lemmas}
\begin{lemma}\label{lem:ninter} %\label{lem:ninter2}
Let $1 \leq x \leq |V|$ and $N \subseteq V_1\subseteq V$. % be such that $x|N|\geq |V_1|$ when $k\geq 2$. % and $v\in V$.
Then
	$$\Exp_{X\subseteq V_1,\; |X|=x}  |  N \cap X  | = \frac{x |N |}{ |V_1|}.$$
%	where the expectation is taken over $x$-subsets $X\subseteq V_1$.
\end{lemma} 
\begin{proof}
%Since the expectation is over $X \subseteq V_1$, %we have
%$$
%	\Exp_X \left[|  N_v\cap X  |^k\right] = \Exp_X |  N_v \cap V_1 \cap X  |.
%$$ 
%So 
%it is sufficient to show that, for any $N \subseteq V_1$,
%\begin{equation}\label{eq:expInstersecLem}
%	\Exp_X \left(|  N\cap X  |\right)^k = \left(\frac{x |N|}{|V_1|}\right)^k,
%\end{equation}
%where the expectation is over $x$-subsets $X\subseteq V_1$. We proceed to the proof.
Let $\mathbf 1_X$ be the indicator function of $X$. Then
observe that the left hand side can be rewritten as
$$
	\Exp_X |  N\cap X  | = \Exp_X \left(\sum_{u \in N} {\mathbf 1}_{X}(u)\right) =
	\sum_{u\in N} \Exp_X  \left({\mathbf 1}_{X}(u)\right) .
$$
%	\left(\sum_{u \in N} \Exp_X  {\mathbf 1}_{X}(u) \right)^k,
%Observe that ${\mathbf 1}_{X}(u)$ and ${\mathbf 1}_{X}(v)$ are independent for $u\neq v$,
%and that $({\mathbf 1}_{X}(u))^l={\mathbf 1}_{X}(u)$ for all $l\geq 1$.
Then we conclude by observing that each term of the sum on the right hand side
 satisfies $\Exp_X  ({\mathbf 1}_{X}(u))=\frac{x}{|V_1|}$, independently of $u\in V_1$.
%This concludes the proof.
\end{proof}

\begin{lemma}\label{ninterExpExp}
Let $1 \leq x,y \leq |V|$.
Let $E(X,Y)$ denote the set of edges between $X$ and $Y$. Then
	$$
		\Exp_{X,Y\subseteq V,\; |X|=x,\; |Y|=y} |  E(X,Y)  | =  \frac{2 x y m}{n^2} ,
	$$
	%and the expectation is over $x$-subsets $X\subseteq V$ and $y$-subsets $Y \subseteq V$, and 
\end{lemma} 
\begin{proof}
For any $v \in V$ we denote $N_v \subseteq V$ its neighbors.
We prove the equality by decomposition the expectation term:
\begin{eqnarray*}
	\Exp_{X,Y} |  E(X,Y)  | 
	&=& \Exp_{X,Y} \sum_{v \in Y} | E(X,\{v\}) | 
	\\
	&=& \Exp_{X} \sum_{v \in V} | E(X,\{v\}) | \times \Pr_{Y}(v \in Y)
	\\
	&=& \frac{y}{n}  \times\sum_{v \in V} \Exp_{X} | N_v \cap X |
	\\
	&=&  \frac{x y}{n^2}  \times\sum_{v \in V}|N_v| \quad\text{by Lemma~\ref{lem:ninter} with $k=1$, $V_1=V$ and $N=N_v$}
	\\
	&=&  \frac{2 x y m}{n^2}.
\end{eqnarray*}
%Hence the conclusion.
\end{proof}

\iffalse
In particular, when the edge degree of the vertices is bounded by, say, $d$, we obtain:
$$
	\Exp_{X} \Exp_{Y} |  E(X,Y)  | = 2 x y \frac{|E|}{|V|^2} \leq x y \frac{d}{V},
$$
since $|E| = \frac{|V| d}{2}$.
\fi

\subsection{Proof of Theorem~\ref{thm:sparsenew}}

Let us denote $\G$ the learning graph of Figure~\ref{SparseNewLG}. It can be seen as a special case of the one of Figure~\ref{LeGallLG}
with $X=\emptyset$ and $A=V$ ({\it i.e.} $x=0$ and $a=n$).
That is we start at Stage~5, and in our case $\Delta(X,B,w)=(N_w \cap B)^2$.
Moreover we are going to use $\Loadall$ everywhere
except for the first part, where we use $\Loadsparse$ in order to minimize the term
$(\mathrm{\bf{S}}_{G,X,A,w})^2$.

Therefore we can duplicate the analysis in the proof of Theorem~\ref{thm:sparse} 
starting from Stage~4 and replacing $\Delta(X,B,w)$ by $(N_w \cap B)$.
Then we get that the negative complexity for any graph $G$  satisfies
$$	C^0(\G, G)
	= O\left(n\left(\frac{b^2 d}{n} \cdot \log(b+1) +\left(\frac{n}{b}\right)^{2}\left(b + \Exp_{w, B} \left( |  N_w\cap B  |^2 \right)\right)\right)\right).$$
	
Then, the last piece of the proof is provided by Lemma~\ref{lem:ninter2} below
which gives, with $x = b, V_1 = V$, and $N = N_w$,
$$
	\Exp_{w\in V, B\subseteq V\,:\, |B|=b} \left( |  N_w\cap B  |^2 \right)
	\leq 2\left( \Exp_w \left(\frac{b^2 |N_w|^2 }{ n^2}\right)\right)
	\leq 2\left( \frac{b^2 (d_2)^2 }{n^2}\right),
$$
where  $d_2 = \sqrt{\Exp_v \big[\,|N_v|^2\,\big]}$.

This concludes the proof of the theorem.

\begin{lemma} \label{lem:ninter2}
Let $1 \leq x \leq |V|$ and $N \subseteq V_1\subseteq V$  be such that $x|N|\geq |V_1|$. % when $k\geq 2$. % and $v\in V$.
Then
	$$\Exp_{X\subseteq V_1,\; |X|=x}  \left(|  N \cap X  |^2 \right) \leq 2\left(\frac{x |N |}{ |V_1|}\right)^2.$$
%	where the expectation is taken over $x$-subsets $X\subseteq V_1$.
\end{lemma} 
\begin{proof}
%Since the expectation is over $X \subseteq V_1$, %we have
%$$
%	\Exp_X \left[|  N_v\cap X  |^k\right] = \Exp_X |  N_v \cap V_1 \cap X  |.
%$$ 
%So 
%it is sufficient to show that, for any $N \subseteq V_1$,
%\begin{equation}\label{eq:expInstersecLem}
%	\Exp_X \left(|  N\cap X  |\right)^k = \left(\frac{x |N|}{|V_1|}\right)^k,
%\end{equation}
%where the expectation is over $x$-subsets $X\subseteq V_1$. We proceed to the proof.
Similarly to the proof of Lemma~\ref{lem:ninter},
let $\mathbf 1_X$ be the indicator function of $X$. Then
%observe that the left hand side can be rewritten as
$$
	\Exp_X\left( |  N\cap X  |^2 \right) = \Exp_X \left( \left(\sum_{u \in N} {\mathbf 1}_{X}(u)\right)^2 \right) =
	\sum_{u,v\in N} \Exp_X \left(  {\mathbf 1}_{X}(u)  {\mathbf 1}_{X}(v)\right).
$$
%	\left(\sum_{u \in N} \Exp_X  {\mathbf 1}_{X}(u) \right)^k,
Observe that ${\mathbf 1}_{X}(u)$ and ${\mathbf 1}_{X}(v)$ are independent for $u\neq v$,
and that $({\mathbf 1}_{X}(u))^2={\mathbf 1}_{X}(u)$.
Therefore
$$
	\Exp_X\left( |  N\cap X  |^2 \right) = \sum_{u, v\in N,\; u\neq v} \left(\Exp_X  {\mathbf 1}_{X}(u)\right) \left(\Exp_X  {\mathbf 1}_{X}(v)\right)+\sum_{u\in N}\Exp_X  {\mathbf 1}_{X}(u).
$$
Remind that for all $u\in V_1$, $\Exp_X  {\mathbf 1}_{X}(u)=\frac{x}{|V_1|}$.
Thus
$$
	\Exp_X\left( |  N\cap X  |^2 \right) = |N|(|N|-1)\left(\frac{x}{|V_1|}\right)^2+|N| \frac{x}{|V_1|}.$$

Using $x|N|\geq |V_1|$, we finally get 
$$
	\Exp_X\left( |  N\cap X  |^2 \right) \leq |N|(|N|-1)\left(\frac{x}{|V_1|}\right)^2+\left(|N| \frac{x}{|V_1|}\right)^2
	\leq 2 \left(|N| \frac{x}{|V_1|}\right)^2.$$

\end{proof}

\end{document}